\documentclass[final]{siamltex}

\usepackage{graphicx}
\usepackage{psfrag}
\usepackage{amsfonts}
\usepackage{amssymb}
\usepackage{amscd}
\usepackage[mathscr]{euscript}
\usepackage{amsmath,latexsym,amsbsy}
\usepackage{enumerate}

\newtheorem{remark}{Remark}[section]
\def\Q{{\mathbb Q}}
\def\R{\mathbb R}
\def\ve{\varepsilon}
\def\se{{\rm se}}
\def\ce{{\rm ce}}

\def\Mc{{\rm Mc}}
\def\Ms{{\rm Ms}}

\title{Localization of Laplacian eigenfunctions in circular, spherical and elliptical domains}

\author{B.-T. Nguyen\footnotemark[1] 
\and D.~S. Grebenkov\footnotemark[1]\ \footnotemark[2] \ \footnotemark[3] \ \footnotemark[4]}

\begin{document}
\maketitle

\renewcommand{\thefootnote}{\fnsymbol{footnote}}
\footnotetext[1]{Laboratoire de Physique de la Mati\`ere Condens\'ee, CNRS -- Ecole Polytechnique, 91128 Palaiseau, France}
\footnotetext[2]{Laboratoire Poncelet, CNRS -- Independent University of Moscow, Bolshoy Vlasyevskiy Pereulok 11, 119002 Moscow, Russia}
\footnotetext[3]{Chebyshev Laboratory, Saint Petersburg State University, 14th line of Vasil'evskiy Ostrov 29, Saint Petersburg, Russia}
\footnotetext[4]{Corresponding author: denis.grebenkov@polytechnique.edu}
\renewcommand{\thefootnote}{\arabic{footnote}}

\date{Received: \today / Revised version: }

\begin{abstract}
We consider Laplacian eigenfunctions in circular, spherical and
elliptical domains in order to discuss three kinds of high-frequency
localization: whispering gallery modes, bouncing ball modes, and
focusing modes.  Although the existence of these modes was known for a
class of convex domains, the separation of variables for above domains
helps to better understand the ``mechanism'' of localization, i.e. how
an eigenfunction is getting distributed in a small region of the
domain, and decays rapidly outside this region.  Using the properties
of Bessel and Mathieu functions, we derive the inequalities which
imply and clearly illustrate localization.  Moreover, we provide an
example of a non-convex domain (an elliptical annulus) for which the
high-frequency localized modes are still present.  At the same time,
we show that there is no localization in most of rectangle-like
domains.  This observation leads us to formulating an open problem of
localization in polygonal domains and, more generally, in piecewise
smooth convex domains.
\end{abstract}




\pagestyle{myheadings}
\thispagestyle{plain}

\section{Introduction}

A hundred years ago, Lord Rayleigh documented an interesting
acoustical phenomenon that occured in the whispering gallery under the
dome of Saint Paul's Cathedral in London \cite{Rayleigh1910} (see also
\cite{Raman1921,Raman1922}).  A whisper of one person propagated along
the curved wall to another person stood near the wall.  This acoustical
effect and many related wave phenomena can be mathematically described
by Laplacian eigenmodes satisfying $-\Delta u = \lambda u$ in a
bounded domain, with an appropriate boundary condition:
\begin{equation}
\begin{split} 
u & = 0  \hskip 5mm  \textrm{(Dirichlet)} , \\
\frac{\partial u}{\partial n} & = 0 \hskip 5mm \textrm{(Neumann)} , \\
\frac{\partial u}{\partial n} + h u & = 0 \hskip 5mm \textrm{(Robin)} , \\
\end{split}
\end{equation}
where $h \geq 0$ is a positive constant, and $\partial/\partial n$ is
the normal derivative directed outwards the boundary.  It turns out
that the eigenmodes that are ``responsible'' for the whispering
effect, are mostly distributed near the boundary of the domain and
almost zero inside.  Keller and Rubinow discussed these so-called
\emph{whispering gallery modes} and also \emph{bouncing ball} modes.
The existence of such \emph{localized} eigenmodes in the limit of
large eigenvalues was shown for any two-dimensional domain with
arbitrary smooth convex curve as its boundary (so-called
high-frequency or high-energy localization) \cite{Keller60}.  A
further semiclassical approximation of Laplacian eigenfunctions in
convex domains was developed by Lazutkin
\cite{Babich68,Lazutkin68,Lazutkin73,Lazutkin93} (see also
\cite{Arnold72,Smith74,Ralston76,Ralston77}).  Chen and
co-workers analyzed Mathieu and modified Mathieu functions and
reported another type of localization named \emph{focusing modes}
\cite{Chen94}.  These and other localized eigenmodes have been
intensively studied for various domains, named quantum billiards
\cite{Gutzwiller,Heller98b,Stockmann,Jakobson01,Sarnak11}.  
It is also worth mentioning that low-frequency localization of
Laplacian eigenfunctions in simple and irregular domains has attracted
a considerable attention during the last two decades
\cite{Sapoval91,Even99,Felix07,Heilman10,Delitsyn12a,Delitsyn12b}.

The aim of this paper consists in revisiting and illustrating the
aforementioned three types of high-frequency localization.  For this
purpose, we consider circular, spherical and elliptical domains for
which the separation of variables reduces the analysis to the behavior
of special functions.  Using the properties of Bessel and Mathieu
functions, we derive the inequalities that clearly show the existence
of infinitely many localized eigenmodes in circular, spherical and
elliptical domains with Dirichlet, Neumann or Robin boundary
condition.  More precisely, we call an eigenfunction $u$ of the
Laplace operator in a bounded domain $\Omega\subset \R^d$
$L_p$-\emph{localized} ($p\geq 1$) if it is essentially supported by a
small subdomain $\Omega_\alpha \subset
\Omega$, i.e.
\begin{equation}
\label{eq:def_loc}
 \frac{\|u\|_{L_p(\Omega \setminus \Omega_\alpha)}}{\|u\|_{L_p(\Omega)}} \ll 1, 
\qquad \frac{\mu_d(\Omega_\alpha)}{\mu_d(\Omega)} \ll 1 ,
\end{equation}
where $\|.\|_{L_p}$ is the $L_p$-norm, and $\mu_d$ is the Lebesgue
measure.  We stress that this ``definition'' is qualitative as there
is no objective criterion for deciding how small these ratios have to
be.  This is the major problem in defining the notion of localization.
For circular, spherical and elliptical domains, we will show in
Sect. \ref{sec:circular} and \ref{sec:elliptical} that both ratios can
be made arbitrarily small.  In other words, for any prescribed
threshold $\ve$, there exist a subdomain $\Omega_\alpha$ and
infinitely many eigenfunctions for which both ratios are smaller than
$\ve$.  Most importantly, we will provide a simple example of a
non-convex domain for which the high-frequency localization is still
present.  At the same time, we will show in Sect. \ref{sec:rectangle}
the absence of localization in most of rectangle-like domains.  This
observation will lead us to formulating an open problem of
localization in polygonal domains and, more generally, in piecewise
smooth convex domains.

\section{Localization in circular and spherical domains}
\label{sec:circular}

\subsection{Eigenfunctions for circular domains}

The rotation symmetry of a disk $\Omega = \{ x\in\R^2 ~:~ |x| < R\}$
of radius $R$ leads to an explicit representation of the
eigenfunctions in polar coordinates:
\begin{equation}
\label{eq:u_disk}
u_{nki}(r,\varphi) = J_n(\alpha_{nk} r/R) 
\times \begin{cases} \cos(n\varphi), \quad i = 1,  \cr  \sin(n\varphi), \quad i = 2 ~ (n\ne 0),\end{cases}
\end{equation}
where $J_n(z)$ are the Bessel functions of the first kind
\cite{Watson,Bowman,Abramowitz} and $\alpha_{nk}$ are the positive
zeros of $J_n(z)$ (Dirichlet), $J'_n(z)$ (Neumann) and $J'_n(z) + h
J_n(z)$ (Robin).  The eigenfunctions are enumerated by the triple
index $nki$, with $n = 0,1,2,...$ counting the order of Bessel
functions, $k = 1,2,3,...$ counting the positive zeros, and $i = 1,2$.
Since $u_{0k2}(r,\varphi)$ are trivially zero, they are not counted as
eigenfunctions.  The eigenvalues $\lambda_{nk} = \alpha_{nk}^2/R^2$,
which are independent of the last index $i$, are simple for $n = 0$
and twice degenerate for $n > 0$.  In the latter case, the
eigenfunction is any nontrivial linear combination of $u_{nk1}$ and
$u_{nk2}$.  As we will derive the estimates that will be independent
of the angular coordinate $\varphi$, the last index $i$ will be
omitted.

\newpage
\subsection{Whispering gallery modes}
\label{sec:whispering}

The disk is the simplest shape for illustrating the whispering gallery
and focusing modes.  The explicit form (\ref{eq:u_disk}) of
eigenfunctions allows one to derive accurate bounds, as shown below.
When the index $k$ is fixed, while $n$ increases, the Bessel functions
$J_n(\alpha_{nk} r/R)$ become strongly attenuated near the origin (as
$J_n(z)\sim (z/2)^n/n!$ at small $z$) and essentially localized near
the boundary, yielding whispering gallery modes.  In turn, when $n$ is
fixed while $k$ increases, the Bessel functions rapidly oscillate, the
amplitude of oscillations decreasing towards to the boundary.  In that
case, the eigenfunctions are mainly localized at the origin, yielding
focusing modes.  These qualitative arguments are rigorously formulated
in the following
\begin{theorem}
\label{theo:disk}
Let $D = \{x \in \R^2 ~:~ |x| < R\}$ be a disk of radius $R > 0$, and
$D_{nk} = \{ x\in\R^2~:~ |x| < R d_n/\alpha_{nk} \}$, where $d_n = n -
n^{2/3}$, and $\alpha_{nk}$ are the positive zeros of $J_n(z)$
(Dirichlet), $J'_n(z)$ (Neumann) or $J'_n(z) + h J_n(z)$ for some $h >
0$ (Robin), with $n = 0,1,2,...$ denoting the order of Bessel function
$J_n(z)$ and $k=1,2,3,...$ counting zeros.  Then for any $p \geq 1$
(including $p = \infty$), there exists a universal constant $C_p > 0$
such that for any $k = 1,2,3,...$ and any large enough $n$, the
Laplacian eigenfunction $u_{nk}$ for Dirichlet, Neumann or Robin
boundary condition satisfies
\begin{equation} 
\label{eq:eigen_disk_L2}
\frac{\|u_{nk}\|_{L_p(D_{nk})}}{\|u_{nk}\|_{L_p(D)}} < C_p n^{\frac{1}{3} + \frac{2}{3p}} 2^{-n^{1/3}/3} .
\end{equation}
This estimate implies that
\begin{equation}
\lim\limits_{n \to \infty} \frac{\|u_{nk}\|_{L_p(D_{nk})}}{\|u_{nk}\|_{L_p(D)}} = 0,
\hskip 10mm \mathrm{while}  \hskip 5mm  \lim\limits_{n \to \infty} {\frac{\mu_2(D_{nk})}{\mu_2(D)}} = 1.
\end{equation}
\end{theorem}
The theorem shows the existence of infinitely many Laplacian
eigenmodes which are $L_p$-localized near the boundary $\partial D$
(see Appendix \ref{sec:Adisk} for a proof).  In fact, for any
prescribed thresholds for both ratios in (\ref{eq:def_loc}), there
exists $n_0$ such that for all $n > n_0$, the eigenfunctions $u_{nk}$
are $L_p$-localized.  These eigenfunctions are called ``whispering
gallery eigenmodes'' and illustrated on Fig. \ref{fig:whispering}.

A simple consequence of the above theorem is 
\begin{corollary}
\label{theo:disk2}
For any $p\geq 1$ and any open subset $V$ compactly included in $D$
(i.e., $\bar{V}\cap \partial D = \emptyset$), one has
\begin{equation}
\label{eq:disk_local}
\lim\limits_{n \to \infty} \frac{\|u_{nk}\|_{L_p(V)}}{\|u_{nk}\|_{L_p(D)}} = 0.
\end{equation}
As a consequence, 
\begin{equation}
\label{eq:CV}
C_p(V) \equiv \inf\limits_{nk} \left\{ \frac{\|u_{nk}\|_{L_p(V)}}{\|u_{nk}\|_{L_p(\Omega)}} \right\} = 0 .
\end{equation}
\end{corollary}
In fact, for any open subset $V$ compactly included in $D$, there
exists $n_0$ such that for all $n > n_0$, $V \subset D_{nk}$ so that
$\|u_{nk}\|_{L_p(V)} \leq \|u_{nk}\|_{L_p(D_{nk})}$ yielding
Eq. (\ref{eq:disk_local}).

\begin{figure}
\begin{center}
\includegraphics[width=120mm]{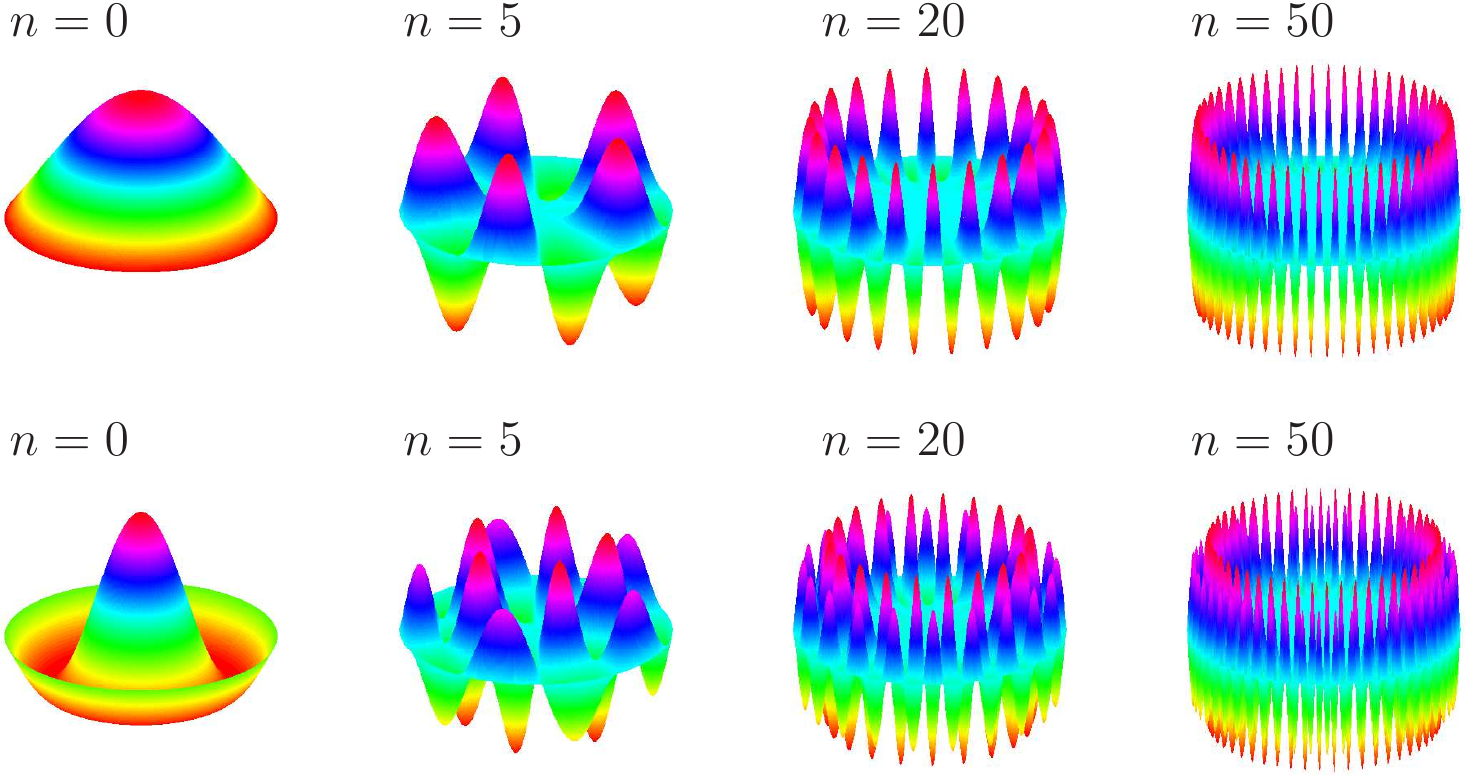}
\end{center}
\caption{
Formation of whispering gallery modes $u_{nk}$ in the unit disk with
Dirichlet boundary condition: for a fixed $k$ ($k = 0$ for top figures
and $k = 1$ for bottom figures), an increase of the index $n$ leads to
stronger localization of the eigenfunction near the boundary. }
\label{fig:whispering}
\end{figure}

In the same way, the localization also happens for any circular
sector.  

In three dimensions, the existence of whispering gallery modes in a
ball follows from the following 
\begin{theorem}
\label{lem:strongsphere1}
\label{theo:ball}
Let $B = \{{x \in \R^3~:~ |x|< R}\}$ be a ball of radius $R$, and
$B_{nk} =\{ x \in B ~:~ 0<|x|< R s_n/\alpha_{nk}\}$, where $s_n = (n +
1/2) - (n + 1/2)^{2/3}$ and $\alpha_{nk}$ are the positive zeros of
$j_n(z)$ (Dirichlet), $j'_n(z)$ (Neumann) or $j'_n(z) + h j_n(z)$ for
some $h > 0$ (Robin), with $n = 0,1,2,...$ denoting the order of the
spherical Bessel function $j_n(z)$ and $k = 1,2,3,...$ counting zeros.
Then, for any $p\geq 1$ (including $p = \infty$), there exists a
universal constant $\tilde{C}_p > 0$ such that for any $k = 1,2,3,...$
and any large enough $n$, the Laplacian eigenfunction $u_{nk}$ with
Dirichler, Neumann or Robin boundary condition satisfies
\begin{equation}
\label{eq:eigen_ball_L2}
\frac{\|u_{nk}\|_{L_p(B_{nk})}}{\|u_{nk}\|_{L_p(B)}} < \tilde{C}_p  (n+1/2)^{\frac{1}{3} + \frac{2}{3p}} 
\exp\left(-\frac{1}{3} \left(n  + \frac{1}{2}\right)^{1/3}\right)    \qquad (n\gg 1).
\end{equation}
As a consequence,
\begin{equation}
\lim\limits_{n \to \infty} \frac{\|u_{nk}\|_{L_p(B_{nk})}}{\|u_{nk}\|_{L_p(B)}} = 0,  
\hskip 10mm  \mathrm{while}  \hskip 5mm   \lim\limits_{n \to \infty} \frac{ \mu_3(B_{nk})}{\mu_3(B)} = 1.
\end{equation}
\end{theorem}
As for the disk, the above results show that infinitely many
high-frequency eigenfunctions are $L_p$-localized near the boundary of
the ball (see Appendix \ref{sec:Aball} for a proof).

\subsection{Focusing modes}

The localization of focusing modes at the origin is described by  
\begin{theorem}
\label{theo:focusingmodes}
For each $R \in (0,1)$, let $D(R) = \{ x\in\R^2 ~:~ R < |x| < 1 \}$,
and $D$ be the unit disk.  Then, for any $n = 0,1,2,...$, the
Laplacian eigenfunction $u_{nk}$ with Dirichlet, Neumann or Robin
boundary condition satisfies
\begin{equation}
\label{eq:focusingNorm}
\lim\limits_{k\to \infty} \frac{\|u_{nk}\|_{L_p(D(R))} }{ \|u_{nk}\|_{L_p(D) }} = \begin{cases} (1 - R^{2-p/2})^{1/p} \quad (1\leq p < 4), \cr
\hskip 10mm  0 \hskip 22mm (p > 4). \end{cases} 
\end{equation}
\end{theorem}
The theorem states that for each non-negative integer $n$, when the
index $k$ increases, the eigenfunctions $u_{nk}$ become more and more
$L_p$-localized near the origin when $p > 4$ (see Appendix
\ref{sec:Adisk} for a proof).  These eigenfunctions are called
``focusing eigenmodes'' and illustrated on Fig. \ref{fig:focusing}.
The theorem shows that the definition of localization is sensitive to
the norm: the focusing modes are $L_p$-localized for $p>4$ (including
$p=\infty$), but they are not $L_p$-localized for $p<4$.  In fact, as
the amplitude of oscillations of the focusing modes exhibits a power
law decay from the origin towards the boundary (see
Fig. \ref{fig:focusing}), the $p$ values control the behavior of the
$L_p$-norm and determine whether the ratio of these norms vanishes or
not in the limit $k\to\infty$.

\begin{figure}
\begin{center}
\includegraphics[width=120mm]{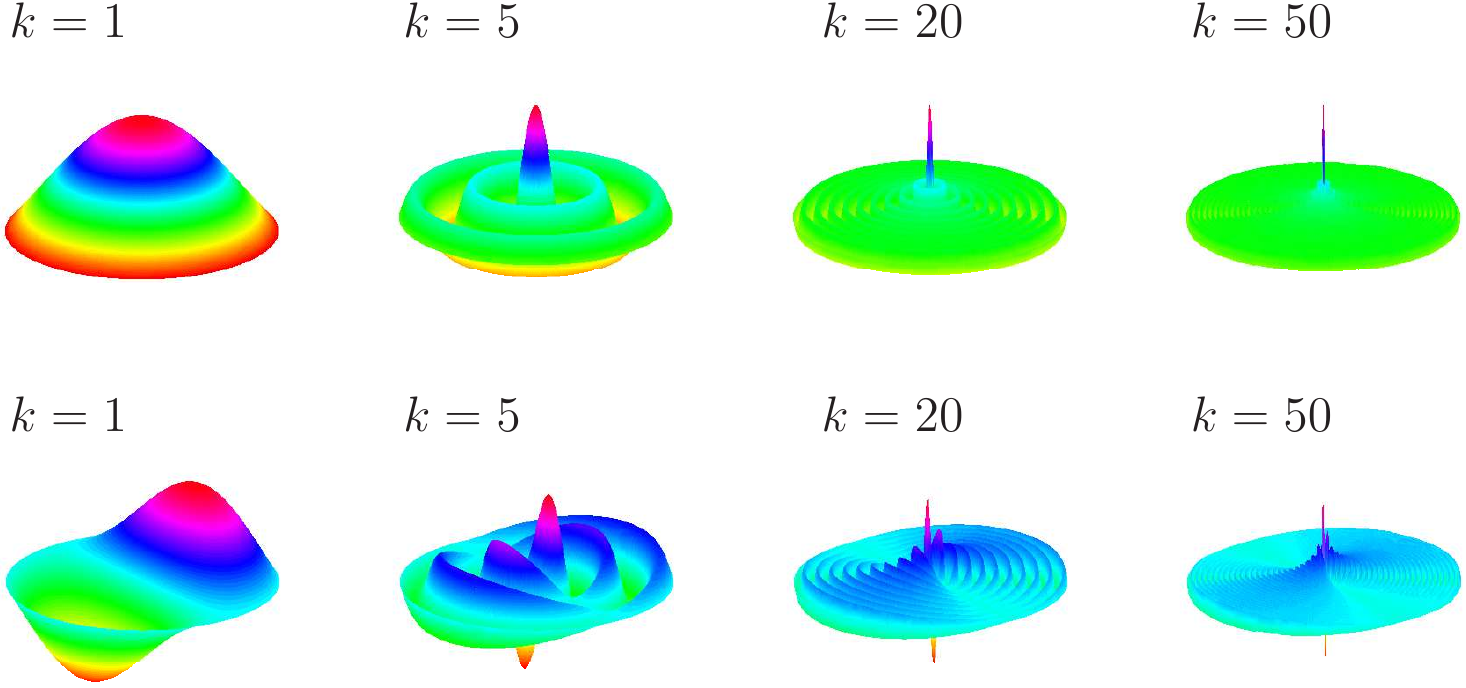}
\end{center}
\caption{
Formation of focusing modes $u_{nk}$ in the unit disk with Dirichlet
boundary condition: for a fixed $n$ ($n = 0$ for top figures and $n =
1$ for bottom figures), an increase of the index $k$ leads to stronger
localization of the eigenfunction at the origin. }
\label{fig:focusing}
\end{figure}

A similar theorem can be reformulated for a ball in three dimensions.
\begin{theorem}
\label{theo:focusingmodes_3d}
For each $R \in (0,1)$, let $B(R) = \{ x\in\R^3 ~:~ R < |x| < 1 \}$,
and $B$ be the unit ball.  Then, for any $n = 0,1,2,...$, the
Laplacian eigenfunction $u_{nk}$ with Dirichlet, Neumann or Robin
boundary condition satisfies
\begin{equation}
\label{eq:focusingNorm_3d}
\lim\limits_{k\to \infty} \frac{\|u_{nk}\|_{L_p(B(R))} }{ \|u_{nk}\|_{L_p(B)}} = \begin{cases} (1 - R^{3-p})^{1/p} \quad (1\leq p < 3), \cr
\hskip 8mm  0 \hskip 22mm (p > 3). \end{cases} 
\end{equation}
\end{theorem}
(see Appendix \ref{sec:Aball} for a proof).

\section{Localization in elliptical domains}
\label{sec:elliptical}

\subsection{Eigenfunctions for elliptical domains}

It is convenient to introduce the elliptic coordinates as
\begin{equation}
\left\{ \begin{array}{ccc}
x_1 &=& a \cosh r \cos \theta , \\
x_2 &=& a \sinh r \sin \theta , \\
\end{array} \right.
\end{equation} 
where $a > 0$ is the prescribed distance between the origin and the
foci, $r \geq 0$ and $0\le \theta < 2\pi$ are the radial and angular
coordinates.  A filled ellipse is a domain with $r < R$ so that its
points $(x_1,x_2)$ satisfy $x_1^2/A^2 + x_2^2/B^2 < 1$, where $R$ is
the elliptic radius and $A = a \cosh R$ and $B = a \sinh R$ are the
major and minor semi-axes.  In the elliptic coordinates, the
separation of the angular and radial variables leads to Mathieu and
modified Mathieu equations, respectively \cite{Chen94,Goldberg90}.
Periodic solutions of the Mathieu equation are possible for specific
characteristic values $c$.  They are denoted as $\ce_n(\theta,q)$ and
$\se_{n+1}(\theta,q)$ (with $n = 0,1,2,...$) and called the angular
Mathieu functions of the first and second kind.  Each function
$\ce_n(\theta,q)$ and $\se_{n+1}(\theta,q)$ corresponds to its own
characteristic value $c$ (the relation being implicit, see
\cite{Mclachlan47}).

For the radial part, there are two linearly independent solutions for
each characteristic value $c$: two modified Mathieu functions
$\Mc_n^{(1)}(r,q)$ and $\Mc_n^{(2)}(r,q)$ correspond to the same $c$
as $\ce_n(\theta,q)$, and two modified Mathieu functions
$\Ms_{n+1}^{(1)}(r,q)$ and $\Ms_{n+1}^{(2)}(r,q)$ correspond to the
same $c$ as $\se_{n+1}(\theta,q)$.  As a consequence, there are four
families of eigenfunctions (distinguished by the index $i=1,2,3,4$) in
a filled ellipse:
\begin{equation}
\begin{split}
u_{nk1} &= \ce_n(\theta,q_{nk1}) \Mc_n^{(1)}(r,q_{nk1}), \\
u_{nk2} &= \ce_n(\theta,q_{nk2}) \Mc_n^{(2)}(r,q_{nk2}), \\
u_{nk3} &= \se_{n+1}(\theta,q_{nk3}) \Ms_{n+1}^{(1)}(r,q_{nk3}), \\ 
u_{nk4} &= \se_{n+1}(\theta,q_{nk4}) \Ms_{n+1}^{(2)}(r,q_{nk4}), \\
\end{split}
\end{equation}
where the parameters $q_{nki}$ are determined by the boundary
condition.  For instance, for a filled ellipse of radius $R$ with
Dirichlet boundary condition, there are four individual equations for
the parameter $q_{nki}$, for each $n = 0,1,2,...$:
\begin{equation}
\label{eq:solveboundaryellipse}
\begin{split}
\Mc_n^{(1)}(R,q_{nk1}) = 0 &,  \quad   \Mc_n^{(2)}(R,q_{nk2}) = 0 , \\
\Ms_{n+1}^{(1)}(R,q_{nk3}) = 0 &,  \quad   \Ms_{n+1}^{(2)}(R,q_{nk4}) = 0 , \\
\end{split}
\end{equation} 
each of them having infinitely many positive solutions $q_{nki}$
enumerated by $k=1,2,\dots$ \cite{Abramowitz,Mclachlan47}.  The
associated eigenvalues $\lambda_{nki}$ are determined as
\begin{equation}
\label{eq:lambda}
\lambda_{nki} = \frac{4q_{nki}}{a^2} .
\end{equation}

The above analysis can be applied almost directly to an elliptical
annulus $\Omega$, i.e. a domain between an inner ellipse $\Gamma_1$
and an outer ellipse $\Gamma_2$, with the same foci.  In elliptic
coordinates, $\Omega$ can be defined by two inequalities: $R_1 < r <
R_2$ and $0\leq \theta < 2\pi$, where the prescribed radii $R_1$ and
$R_2$ determine $\Gamma_1$ and $\Gamma_2$, respectively.

We consider two families of eigenfunctions in $\Omega$:
\begin{equation}
\begin{split}
u_{nk1} &= \ce_n(\theta,q_{nk1})\left[{a_{nk1}\Mc_n^{(1)}(r,q_{nk1}) + b_{nk1} \Mc_n^{(2)}(r,q_{nk1})}\right] ,\\
u_{nk2} &= \se_{n+1}(\theta,q_{nk2})\left[{a_{nk2}\Ms_{n+1}^{(1)}(r,q_{nk2}) + b_{nk2} \Ms_{n+1}^{(2)}(r,q_{nk2})}\right] .\\
\end{split}
\end{equation}
The parameters $a_{nki}$, $b_{nki}$ and $q_{nki}$ ($i=1,2$) are set by
boundary conditions and the normalization of eigenfunctions.  For
Dirichlet boudary condition, one solves the following equations:
\begin{equation} 
\label{eq:annuliBound}
\begin{split}
\Mc_n^{(1)}(R_1,q_{nk1}) \Mc_n^{(2)}(R_2,q_{nk1}) - \Mc_n^{(1)}(R_2,q_{nk1}) \Mc_n^{(2)}(R_1,q_{nk1}) &= 0, \\
\Ms_{n+1}^{(1)}(R_1,q_{nk2}) \Ms_{n+1}^{(2)}(R_2,q_{nk2}) - \Ms_{n+1}^{(1)}(R_2,q_{nk2}) \Ms_{n+1}^{(2)}(R_1,q_{nk2}) &= 0. \\
\end{split}
\end{equation}
For $n = 0,1,2,...$, each of these equations has infinitely many
solutions $q_{nki}$ enumerated by $k = 1,2,3,...$ \cite{Mclachlan47}.
The eigenvalues are determined by Eq. (\ref{eq:lambda}).

\newpage
\subsection{Bouncing ball modes}
\label{sec:bouncing}

For each $\alpha \in \left(0,\frac{\pi}{2}\right)$, we consider the
elliptical sector $\Omega_\alpha$ inside an elliptical domain $\Omega$:
\begin{equation*}
\Omega_\alpha =  \left\{ R_1 < r < R_2,~ \theta \in (\alpha,\pi-\alpha) \cup (\pi+\alpha,2\pi-\alpha)\right\}.
\end{equation*}
\begin{theorem}
\label{theo:localizationEllipticalAnnuliMathieuEigenf}
Let $\Omega$ be a filled ellipse or an elliptical annulus (with a
focal distance $a > 0$).  For any $\alpha \in
\left(0,\frac{\pi}{2}\right)$, $p \geq 1$ and $i=1,2,3,4$ (for
filled ellipse) or $i=1,2$ (for elliptical annulus), there exists
$\Lambda_\alpha >0 $ such that for any $\lambda_{nki} >
\Lambda_\alpha$
\begin{equation}
\label{eq:bound_elliptical}
\frac{\left\|u_{nki}\right\|_{L_p(\Omega\setminus \Omega_\alpha)}}{\left\|u_{nki}\right\|_{L_p(\Omega)}} < 
D_n \left(\frac{16\alpha}{\pi-\alpha/2}\right)^{1/p}
\exp\left(-a\sqrt{\lambda_{nki}} \left[\sin\left(\frac{\pi}{4}+\frac{\alpha}{2}\right) - \sin\alpha\right]\right),
\end{equation}
where
\begin{equation}
D_n = 3\sqrt{\frac{1+\sin\left({\frac{3\pi}{8}+\frac{\alpha}{4}}\right)}
{\bigl[\tan\left({\frac{\pi}{16}-\frac{\alpha}{8}}\right)\bigr]^n}} .
\end{equation}
\end{theorem}
(see Appendix \ref{sec:Aelliptical} for a proof; a similar
exponential bound for the $L_2$-norm was recently derived in
\cite{Betcke11}).  Given that $\lambda_{nki}\to\infty$ as $k$
increases (for any fixed $n$ and $i$), while the area of
$\Omega_\alpha$ can be made arbitrarily small by sending $\alpha\to
\pi/2$, the theorem implies that there are infinitely many
eigenfunctions $u_{nki}$ which are $L_p$-localized in the elliptical
sector $\Omega_\alpha$:
\begin{equation}
\lim\limits_{k \to \infty} \frac{\|u_{nki}\|_{L_p(\Omega\setminus \Omega_\alpha)}}
{\|u_{nki}\|_{L_p(\Omega)}} = 0. 
\end{equation}
These eigenfunctions are called ``bouncing ball modes'' and
illustrated on Fig. \ref{fig:eigen_bouncing}.  We note that similar
results were already known for a filled ellipse and, more generally,
for convex planar domains with smooth boundary \cite{Keller60,Chen94}.
Although our estimates are specific to elliptical shapes, they are
explicit, simpler and also applicable to $L_p$ norms and to elliptical
annuli, i.e. non-convex domains.

The quality of the above estimates was checked numerically.  Figure
\ref{fig:TestTheorem1Alpha1} shows the ratio
$\frac{\|u_{nk1}\|_{L_2(\Omega\setminus\Omega_\alpha)}}{\|u_{nk1}\|_{L_2(\Omega_\alpha)}}$
and its upper bound for two families of eigenfunctions in a filled
ellipse and an elliptical annulus.  One can clearly see the rapid
exponential decay of this ratio when $k$ increases that implies the
localization in a thin sector around the vertical (minor) axis.  Note
that the upper bound is not sharp and can be further improved.

\begin{figure}
\begin{center}
\includegraphics[width=120mm]{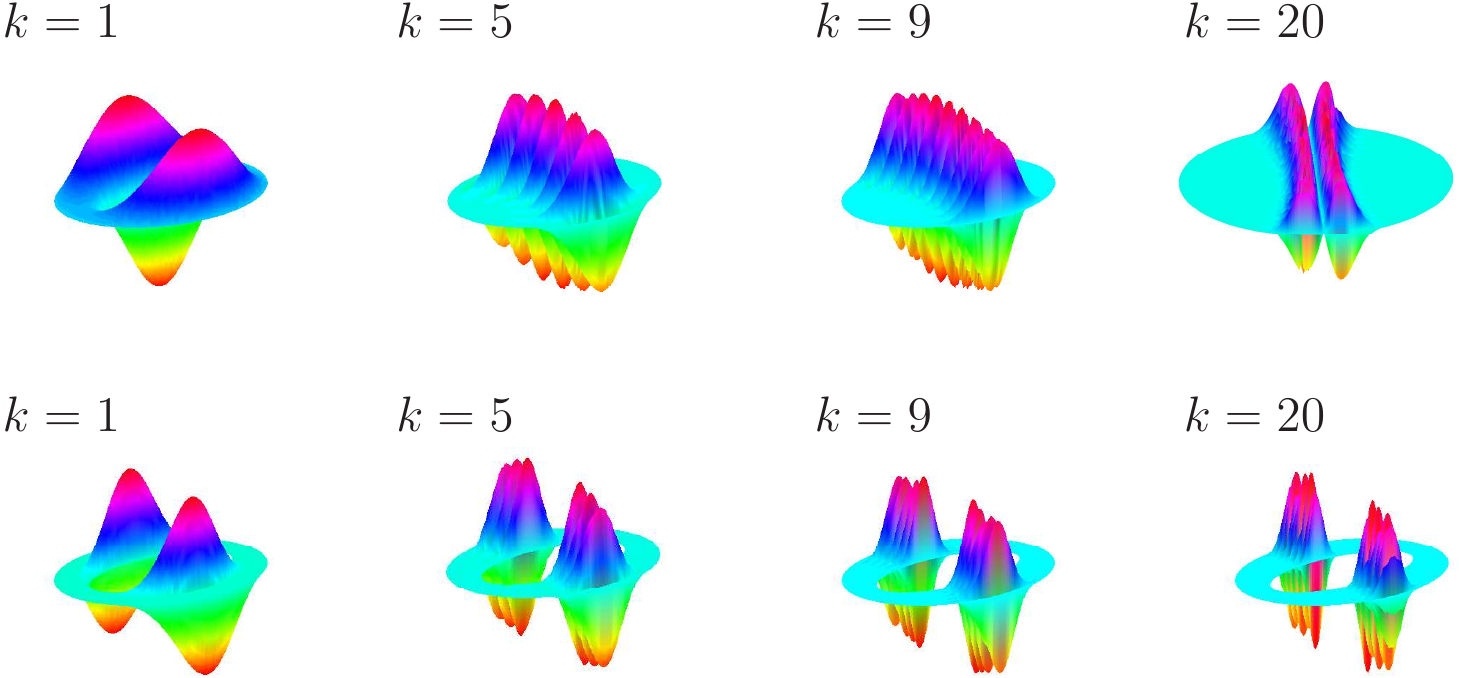}
\end{center}
\caption{
Formation of bouncing ball modes $u_{nki}$ in a filled ellipse of
radius $R=1$ (top) and an elliptical annulus of radii $0.5$ and $1$
(bottom), with the focal distance $a = 1$ and Dirichlet boundary
condition.  For fixed $n = 1$ and $i = 1$, an increase of the index
$k$ leads to stronger localization of the eigenfunction near the
vertical semi-axis ($K_{max} = 200$, see Appendix
\ref{sec:CalMathieu}). }
\label{fig:eigen_bouncing}
\end{figure}

\begin{figure}
 \begin{center}
\includegraphics[width=60mm]{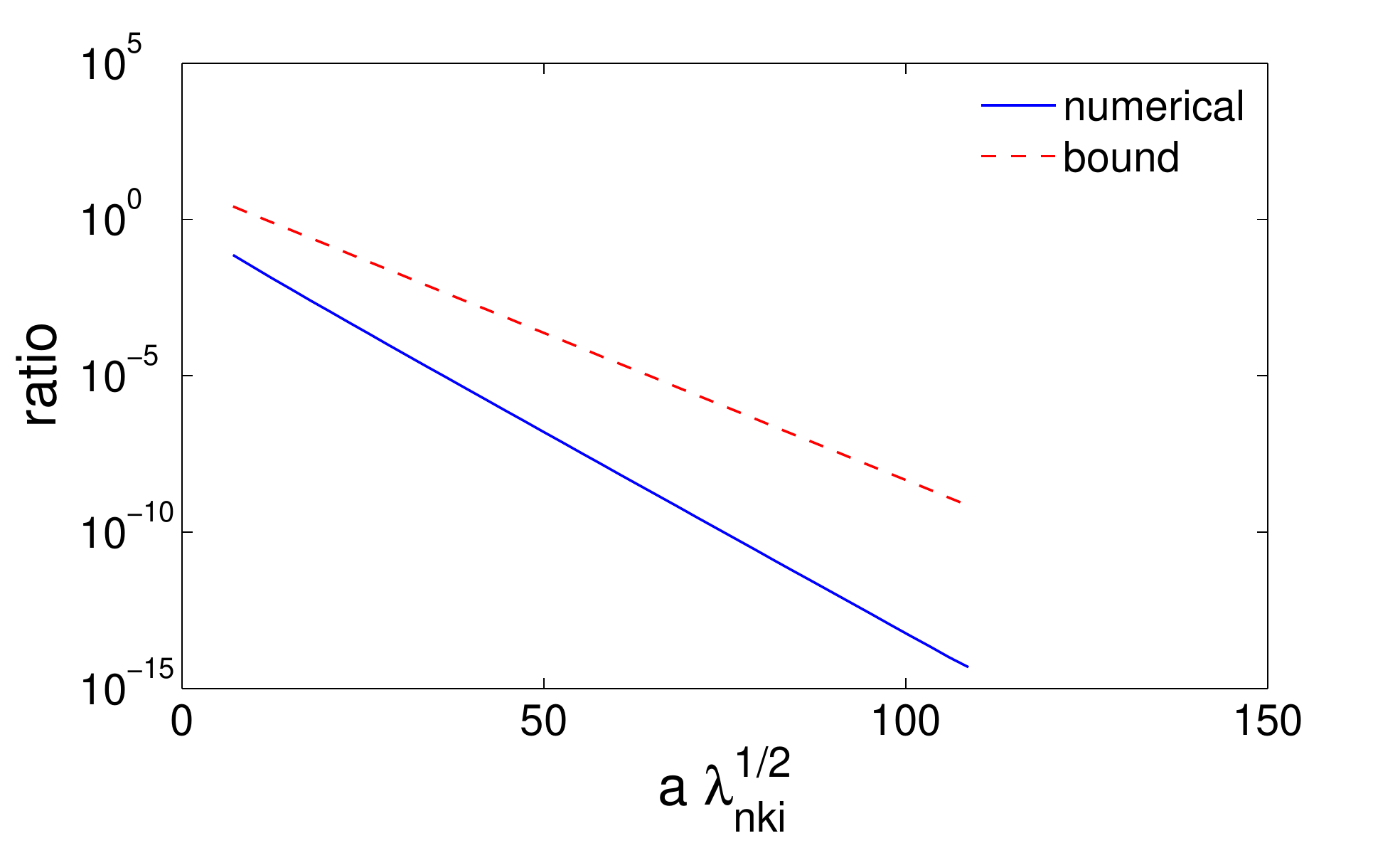}
\includegraphics[width=60mm]{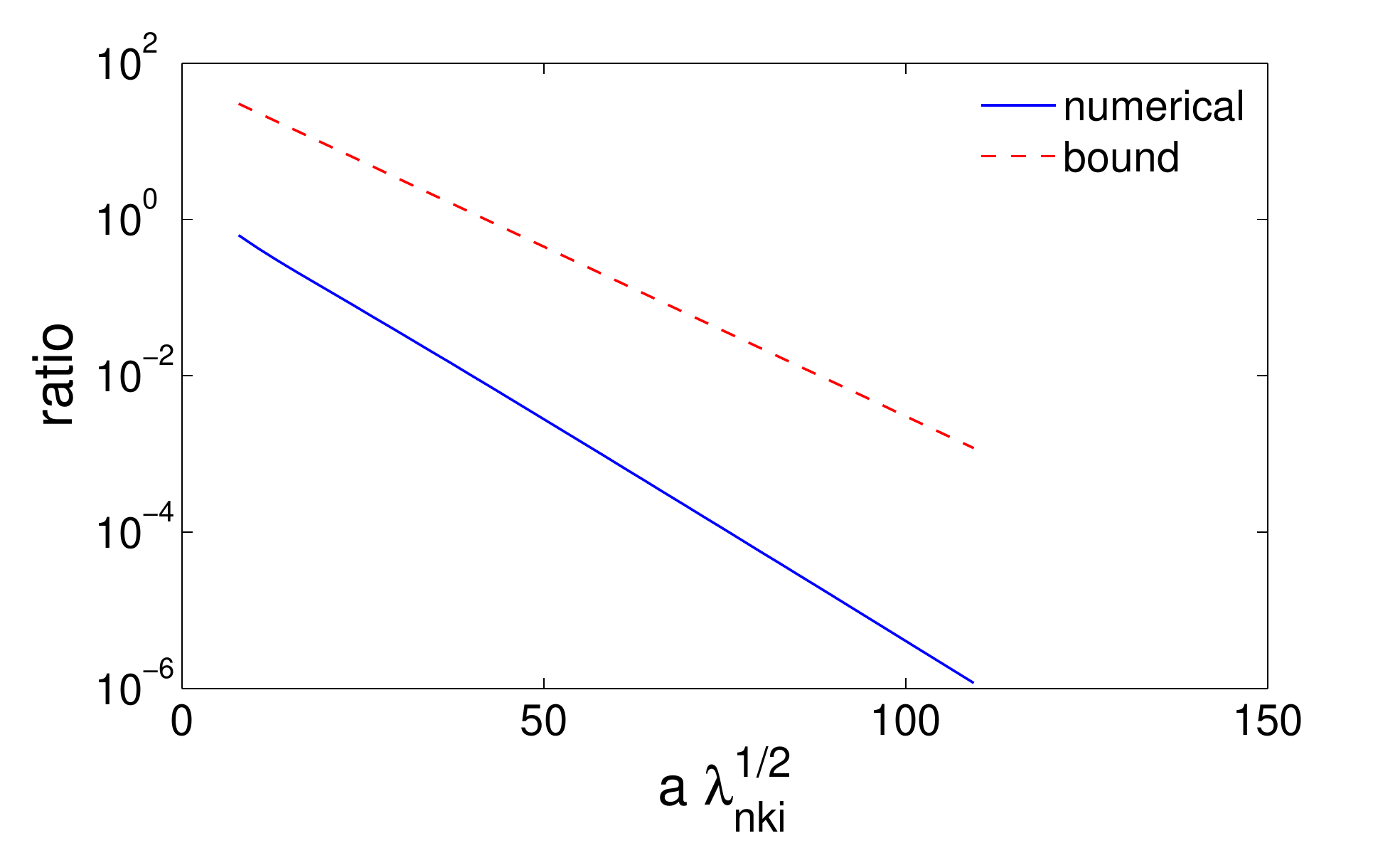}
\includegraphics[width=60mm]{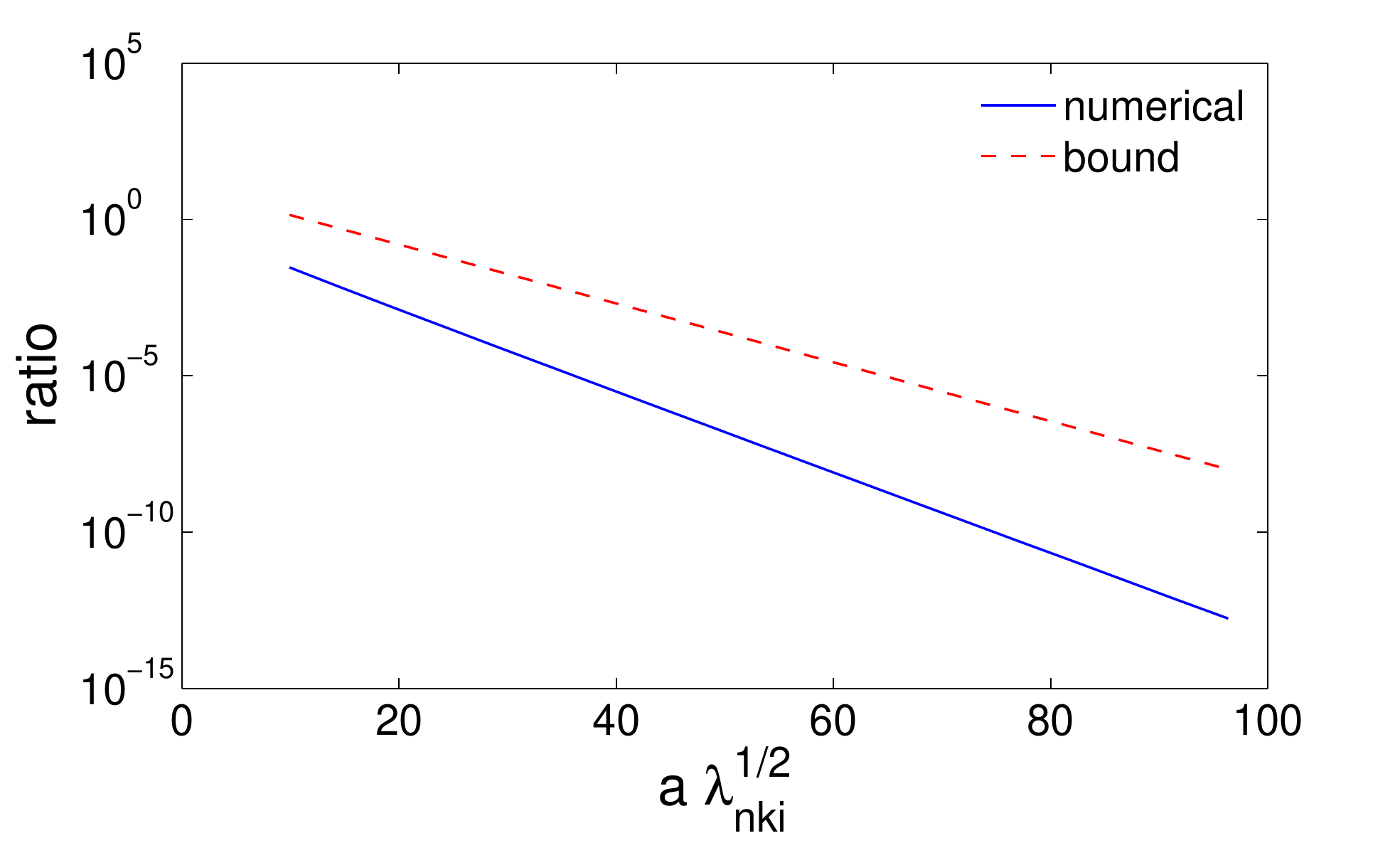}
\includegraphics[width=60mm]{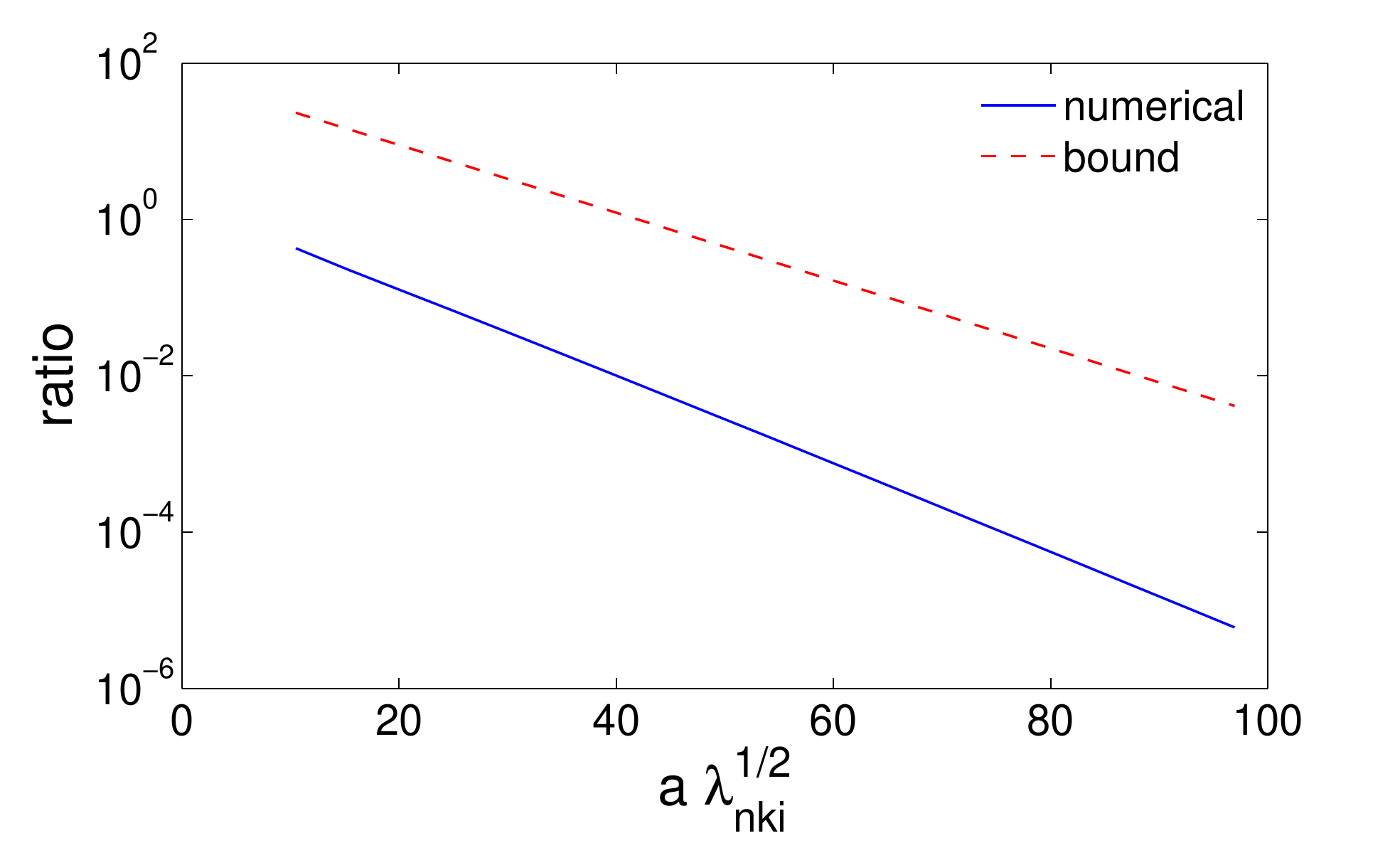}
\end{center}
\caption{
The ratio
$\frac{\|u_{nk1}\|_{L_2(\Omega\setminus\Omega_\alpha)}}{\|u_{nk1}\|_{L_2(\Omega_\alpha)}}$
(solid blue line) and its upper bound (\ref{eq:bound_elliptical})
(dashed red line) in a filled ellipse of radius $R = 1$ and focal
distance $a = 1$ (top) and in an elliptical annulus of radii $R =
0.5$ and $R = 1$ and focal distance $a = 1$ (bottom), with $n = 0$ and
$\alpha = \pi/4$ (left) and $n = 1$ and $\alpha = \pi/3$ (right).  For
numerical computation of Mathieu functions, we used $K_{max}= 200$
(see Appendix \ref{sec:CalMathieu}).}
\label{fig:TestTheorem1Alpha1}
\end{figure}

\section{Discussion}
\label{sec:discussion}

The explicit estimates from previous sections provide us with simple
examples of domains for which there are infinitely many
$L_p$-localized eigenfunctions, according to the definition
(\ref{eq:def_loc}).  Most importantly, the high-frequency localization
may occur in both convex and non-convex domains.  This observation
relaxes, at least for elliptical domains, the condition of convexity
that was significant for the construction of whispering gallery and
bouncing ball modes by Keller and Rubinow \cite{Keller60} and for
semiclassical approximations by Lazutkin
\cite{Lazutkin73,Lazutkin93}.  At the same time, these
approximations suggest the existence of $L_p$-localized eigenfunctions
for a large class of domains.  How large is this class?  What are the
relevant conditions on the domain?  To our knowledge, these questions
are open.  In order to highlight the relevance of these questions, it
is instructive to give an example of domains for which there is no
localization.

\subsection{Rectangle-like domains}
\label{sec:rectangle}

Rectangle-like domains, $\Omega = (0,\ell_1)\times ...\times
(0,\ell_d)\subset \R^d$ (with the sizes $\ell_i > 0$), may seem to
present the simplest shape for studying the Laplacian eigenfunctions
as they are factored and expressed through sines (Dirichlet), cosines
(Neumann) or their combination (Robin):
\begin{equation}
\begin{split}
u_{n_1,...,n_d}(x_1,...,x_d) & = u_{n_1}^{(1)}(x_1) \ldots u_{n_d}^{(d)}(x_d) , \qquad
\lambda_{n_1,...,n_d} = \lambda_{n_1}^{(1)} + ... + \lambda_{n_d}^{(d)} , \\
\end{split}
\end{equation}
with the multiple index $n_1...n_d$, and $u_{n_i}^{(i)}(x_i)$ and
$\lambda_{n_i}^{(i)}$ ($i=1,...,d$) corresponding to the
one-dimensional problem on the interval $(0,\ell_i)$:
\begin{equation*}
\begin{split}
u_n^{(i)}(x) & = \sin(\pi n x/\ell_i) , \quad \lambda_n^{(i)} = \pi^2 n^2/\ell_i^2, 
\quad n = 1,2,3,... \hskip 10mm  \rm{(Dirichlet)} , \\
u_n^{(i)}(x) & = \cos(\pi n x/\ell_i) , \quad \lambda_n^{(i)} = \pi^2 n^2/\ell_i^2, 
\quad n = 0,1,2,...  \hskip 10mm \rm{(Neumann)}  \\
\end{split}
\end{equation*}
(Robin boundary condition will not be considered here).  The situation
is indeed elementary for rectangle-like domains for which all
eigenvalues are simple.
\begin{theorem}
\label{theo:rectangle}
Let $\Omega = (0,\ell_1)\times ... \times (0,\ell_d) \subset \R^d$ be
a rectangle-like domain with sizes $\ell_1 > 0$, ..., $\ell_d > 0$
such that
\begin{equation}
\label{eq:rect_cond}
\ell_i^2/\ell_j^2 \notin \Q    \qquad \forall~ i\ne j .
\end{equation}
($\Q$ denoting the set of rational numbers).  Then for any $p\geq 1$
and any open subset $V\subset \Omega$,
\begin{equation}
\label{eq:C2V}
C_p(V) = \inf\limits_{n_1,...,n_d} \left\{ \frac{\|u_{n_1,...,n_d}\|_{L_p(V)}}{\|u_{n_1,...,n_d}\|_{L_p(\Omega)}} \right\} > 0 .
\end{equation}
\end{theorem}
The proof is elementary (see Appendix \ref{sec:Arectangle}) and relies
the fact that all the eigenvalues are simple due to the condition
(\ref{eq:rect_cond}).  The fact that $C_p(V) > 0$ for any open subset
$V$ means that there is no eigenfunction that could fully ``avoid''
any location inside the domain, i.e., there is no $L_p$-localized
eigenfunction.  Since the set of rational numbers has zero Lebesgue
measure, the condition (\ref{eq:rect_cond}) is fulfilled almost
surely, if one would choose a rectangle-like domain randomly.  In
other words, for most rectangle-like domains, there is no
$L_p$-localized eigenfunction.

When at least one ratio $\ell_i^2/\ell_j^2$ is rational, certain
eigenvalues are degenerate, and the associate eigenfunctions become
linear combinations of products of sines or cosines.  For instance,
for the square with $\ell_1 = \ell_2 = \pi$ and Dirichlet boundary
condition, the eigenvalue $\lambda_{1,2} = 1^2 + 2^2$ is twice
degenerate, and $u_{1,2}(x_1,x_2) = c_1 \sin(x_1)\sin(2x_2) + c_2
\sin(2x_1)\sin(x_2)$, with arbitrary constants $c_1$ and $c_2$
($c_1^2+c_2^2 \ne 0$).  Although the computation is still elementary
for each eigenfunction, it is unknown whether the infimum $C_p(V)$
from Eq. (\ref{eq:C2V}) is strictly positive or not, for arbitrary
rectangle-like domain $\Omega$ and any open subset $V$.  The most
general known result for a rectangle $\Omega = (0,\ell_1)\times
(0,\ell_2)$ states that $C_2(V) > 0$ for any $V\subset \Omega$ of the
form $V = (0,\ell_1)\times \omega$, where $\omega$ is any open subset
of $(0,\ell_2)$ \cite{Burq05}.  Even for the unit square, the
statement $C_p(V) > 0$ for any open subset $V$ seems to be an open
problem.  More generally, one may wonder whether $C_p(V)$ is strictly
positive or not for any open subset $V$ in polygonal convex domains or
in piecewise smooth convex domains.  To our knowledge, these questions
are open.

\section{Conclusion}

We revived the classical problem of high-frequency localization of
Laplacian eigenfunctions.  For circular, spherical and elliptical
domains, we derived the inequalities for $L_p$-norms of the Laplacian
eigenfunctions that clearly illustrate the emergence of whispering
gallery, bouncing ball and focusing eigenmodes.  We gave an
alternative proof for the existence of bouncing ball modes in
elliptical domains.  This proof relies on the properties of Mathieu
functions and is as well applicable to elliptical annuli.  As a
consequence, bouncing ball modes also exist in non-convex domains.  At
the same time, we showed that there is no localization in most
rectangle-like domains that led us to formulating the problem of how
to characterize the class of domains admitting high-frequency
localization.  In particular, the roles of convexity and smoothness
have to be further investigated.  The problem of localization in
polygonal convex domains or, more generally, in piecewise smooth
convex domains are open.

\section*{Acknowledgments}

The authors acknowledge helpful discussions with Ya. G. Sinai and
S. Nonnenmacher.

\appendix
\section{Proofs for a disk}
\label{sec:Adisk}

The proof of Theorem \ref{theo:disk} is based on several estimates for
Bessel functions and their roots that we recall in the following
lemmas.  In this Appendix, $j_{\nu,k}$ and $j'_{\nu,k}$ denote all
positive zeros (enumerated by $k=1,2,3,...$ in an increasing order) of
the Bessel function $J_\nu(x)$ and its derivative $J'_\nu(x)$,
respectively.
\begin{lemma}
\label{lem:Kroger1}
For any $n = 1,2,3,...$ and any $\varepsilon \in (0,2/3)$, the Bessel
function $J_n(x)$ satisfies \cite{Kroger96}
\begin{equation}
\label{eq:disk_lemma1}  
0 < J_n(nz) < 2^{-n^{\varepsilon}/3}  \qquad \forall~ z \in (0, 1 - n^{\varepsilon - \frac23}) .
\end{equation}
\end{lemma}

\begin{lemma}
\label{lem:Chambers}
The first zeros $j_{n,1}$ and $j'_{n,1}$ with $n = 1,2,...$ satisfy
\cite{Watson,Chambers82}
\begin{equation}
\label{eq:disk_lemma2}  
n < j'_{n,1} < j_{n,1} < \sqrt{n+1} \left({\sqrt{n+2}+1}\right).
\end{equation} 
\end{lemma}

\begin{lemma}
\label{lem:root_estim}
For large enough $n$, the asymptotic relations hold \cite{Watson}:
\begin{eqnarray}
J_n(n) &=& C'_1 n^{-1/3} + O(n^{-5/3})   \qquad \left(C'_1 = \frac{\Gamma(1/3)}{2^{2/3} 3^{1/6} \pi} \approx 0.4473\right) \\
j'_{n,1} &=& n + n^{1/3} ~ C'_2 + O(n^{-1/3})    \qquad (C'_2 = 0.808618...) .
\end{eqnarray}
As a consequence, taking smaller constants (e.g., $C_1 = 0.447$ and
$C_2 = 0.8086$), one gets lower bounds for large enough $n$:
\begin{eqnarray}
\label{eq:Jn(n)}
J_n(n) &>& C_1 n^{-1/3}  \qquad (n\gg 1) \\
\label{eq:jprime_n}
j'_{n,1} &>& n + C_2 n^{1/3}  \qquad (n\gg 1) ,
\end{eqnarray}
\end{lemma}

\begin{lemma}
\label{lem:Olver}
For fixed $k$ and large $\nu$, the Olver's expansion holds
\cite{Olver51,Olver52,Elbert01}
\begin{equation}
\begin{split}
j_{\nu,k} & = \nu + \delta_k \nu^{1/3} + \frac{3}{10} \delta_k^2 \nu^{-1/3} + \frac{5-\delta_k^3}{350} \nu^{-1}
- \frac{479 \delta_k^4 + 20\delta_k}{63000} \nu^{-5/3} \\
& + \frac{20231 \delta_k^5 - 27550 \delta_k^2}{8 085 000} \nu^{-7/3} + O(\nu^{-3}) , \\
\end{split}
\end{equation}
where $\delta_k = -a_k 2^{-1/3} > 0$ and $a_k$ are the negative zeros
of the Airy function (e.g., $\delta_1 = 1.855757...$).  Taking $c_k =
\delta_k + \epsilon$ (e.g., $\epsilon = 1$), one gets the upper bounds
for $j_{\nu,k}$ for $\nu$ large enough
\begin{equation}
\label{eq:j_nuk}
j_{\nu,k} < \nu + c_k \nu^{1/3}  \qquad (\nu \gg 1).
\end{equation}
\end{lemma}

\begin{lemma}
For fixed $\nu$ and large $k$, the McMahon's expansion holds
\cite{Watson} (p. 506)
\begin{equation}
j_{\nu,k} = k\pi + \frac{\pi}{2}(\nu - 1/2) - \frac{4\nu^2 - 1}{8(k\pi + \pi(\nu-1/2)/2)} + O(1/k^3) .
\end{equation}
\end{lemma}

\begin{lemma}
\label{lem:maxima}
The absolute extrema of any Bessel function $J_\nu(z)$ progressively
decrease \cite{Watson} (p. 488), i.e.
\begin{equation}
|J_\nu(j'_{\nu,1})| > |J_\nu(j'_{\nu,2})| > |J_\nu(j'_{\nu,3})| > ...
\end{equation}
\end{lemma}

\begin{lemma}
\label{lem:Robin}
The $k$-th positive zero $\alpha_{nk}$ of the function $J'_n(z) + h
J_n(z)$ for any $h > 0$ lies between the $k$-th positive zeros
$j_{n,k}$ and $j'_{n,k}$ of the Bessel function $J_n(z)$ and its
derivative $J'_n(z)$:
\begin{equation}
j'_{n,k} < \alpha_{nk} < j_{n,k}.
\end{equation}
\end{lemma}
\begin{proof} 
This is a direct consequence of the minimax principle that ensures the
monotonous increase of eigenvalues with the parameter $h$
\cite{Courant}.
\end{proof}

Using these lemmas, we prove Theorem \ref{theo:disk}.
\begin{proof}
The proof formalizes the idea that the eigenfunction $u_{nk}$ is small
in the large subdomain $D_{nk} = \{ x\in D ~:~ |x| < R
d_n/\alpha_{nk}\}$ (with $d_n = n - n^{2/3}$) and large in the small
subdomain $A_{nk} = \{ x\in D ~:~R n/\alpha_{nk} < |x| < R
j'_{n,1}/\alpha_{nk}\}$.  Since $A_{nk}\subset D$, we have for $1\leq
p < \infty$
\begin{equation*}
\frac{\|u_{nk}\|_{L_p(D_{nk})}^p}{\|u_{nk}\|_{L_p(D)}^p} < \frac{\|u_{nk}\|_{L_p(D_{nk})}^p}{\|u_{nk}\|_{L_p(A_{nk})}^p} = 
\frac{\int\limits_0^{Rd_n/\alpha_{nk}} dr ~r~ |J_n(r \alpha_{nk}/R)|^p}
{\int\limits_{R n/\alpha_{nk}}^{R j'_{n,1}/\alpha_{nk}} dr~ r~ |J_n(r \alpha_{nk}/R)|^p} = 
\frac{\int\limits_0^{d_n} dz z~ |J_n(z)|^p}{\int\limits_{n}^{j'_{n,1}} dz z~ |J_n(z)|^p} .
\end{equation*}
The numerator can be bounded by the inequality (\ref{eq:disk_lemma1})
with $\epsilon = 1/3$:
\begin{equation*}
\int\limits_{0}^{d_n} dz z~ |J_n(z)|^p < \left(2^{-n^{1/3}/3}\right)^p ~ \frac{d_n^2}{2} <  2^{-p n^{1/3}/3} ~ \frac{n^2}{2} 
\qquad (n=1,2,3,...).
\end{equation*}
In order to bound the denominator, we use the inequalities
(\ref{eq:Jn(n)}, \ref{eq:jprime_n}) and the fact that $J_n(z)$
increases on the interval $[n,j'_{n,1}]$ (up to the first maximum at
$j'_{n,1}$):
\begin{equation*}
\int\limits_{n}^{j'_{n,1}} dz z~ |J_n(z)|^p > |J_n(n)|^p \frac{(j'_{n,1})^2 - n^2}{2} 
>  [C_1 n^{-1/3}]^p ~ \frac{(n+C_2n^{1/3})^2 - n^2}{2} >  C_1^p C_2 n^{(4-p)/3}  
\end{equation*}
for $n$ large enough, from which
\begin{equation*}
\frac{\|u_{nk}\|_{L_p(D_{nk})}}{\|u_{nk}\|_{L_p(A_{nk})}} < \frac{n^{\frac13 + \frac{2}{3p}} 2^{-n^{1/3}/3}}{C_1 (2C_2)^{1/p}} \qquad (n\gg 1)
\end{equation*}
that implies Eq. (\ref{eq:eigen_disk_L2}).

For $p = \infty$, one has
\begin{equation*}
\frac{\|u_{nk}\|_{L_\infty(D_{nk})}}{\|u_{nk}\|_{L_\infty(D)}} < \frac{\|u_{nk}\|_{L_\infty(D_{nk})}}{\|u_{nk}\|_{L_\infty(A_{nk})}} = 
\frac{\max\limits_{0< z < d_n} |J_n(z)|}{\max\limits_{n < z < j'_{n,1}} |J_n(z)|} .
\end{equation*}
Using the same bounds as above, one gets
\begin{equation*}
\max\limits_{0< z < d_n} |J_n(z)| < 2^{-n^{1/3}/3},  \qquad
\max\limits_{n < z < j'_{n,1}} |J_n(z)| > J_n(n) > C_1 n^{-1/3} 
\end{equation*}
that implies Eq. (\ref{eq:eigen_disk_L2}).

Finally, from Lemmas \ref{lem:Olver} and \ref{lem:Robin}, we have
\begin{equation*}
1 > \frac{\mu_2(D_{nk})}{\mu_2(D)} = \left(\frac{d_n}{\alpha_{nk}}\right)^2 > \frac{d_n^2}{j_{n,k}^2} > 
\frac{(n - n^{2/3})^2}{(n + c_k n^{1/3})^2}  \qquad (n\gg 1)
\end{equation*}
so that the ratio of the areas tends to $1$ as $n$ goes to infinity.
\end{proof}

Let us prove Theorem \ref{theo:focusingmodes}.
\begin{proof}
For $p = \infty$, the explicit representation (\ref{eq:u_disk}) of
eigenfunctions leads to
\begin{equation}
\label{eq:A_auxil1}
 \frac{\|u_{nk}\|_{L_{\infty}(D(R))} }{ \|u_{nk}\|_{L_{\infty}(D)} }  
= \frac{\max\limits_{r\in[R,1]}{\left|J_n(\alpha_{nk} r)\right|}}{\max\limits_{r\in[0,1]}{\left|J_n(\alpha_{nk} r)\right|}}  
= \frac{\max\limits_{r\in[R,1]}{\left|J_n(\alpha_{nk} r)\right|}}{\left|J_n(j'_{n,1})\right|} ,
\end{equation}
where we used the fact that the first maximum (at $j'_{n,1}$) is the
largest (Lemma \ref{lem:maxima}).  Since $\lim\limits_{k\to \infty}
{\alpha_{nk}} = \infty$, the Bessel function $J_n(\alpha_{nk} r)$ with
$k\gg 1$ can be approximated in the interval $[R,1]$ as
\cite{Abramowitz}
\begin{equation}
\label{eq:BesselapproxOne} 
J_n(\alpha_{nk} r) \approx \sqrt{\frac{2}{\pi \alpha_{nk} r}} \cos\left(\alpha_{nk} r - \frac{n\pi}{2} - \frac{\pi}{4}\right) .
\end{equation}
It means that there exists a positive integer $K_0$ and a constant
$A_0>0$ (e.g., $A_0 = 3/\pi$) such that
\begin{equation}
\label{eq:Bessel_ineq}
\left|{J_n(\alpha_{nk} r)}\right| < \sqrt{\frac{A_0}{\alpha_{nk} r}} \leq \sqrt{\frac{A_0}{\alpha_{nk} R}}, \quad \forall r\in[R,1], ~ k > K_0.
\end{equation}
Given that the denominator in Eq. (\ref{eq:A_auxil1}) is fixed, while
the numerator decays as $\alpha_{nk}^{-1/2}$, one gets
Eq. (\ref{eq:focusingNorm}) for $p = \infty$.  

For $p > 4$, the ratio of $L_p$ norms is 
\begin{equation}
\frac{\| u_{nk} \|_{L_p(D(R))}^p}{\| u_{nk} \|_{L_p(D)}^p} = 
\frac{\int\limits_{\alpha_{nk} R}^{\alpha_{nk}} dr ~ r ~ |J_n(r)|^p}{\int\limits_{0}^{\alpha_{nk}} dr ~ r ~ |J_n(r)|^p} .
\end{equation}
The inequality (\ref{eq:Bessel_ineq}) allows one to bound the
numerator as
\begin{equation*}
\int\limits_{\alpha_{nk} R}^{\alpha_{nk}} dr ~ r ~ |J_n(r)|^p \leq A_0^{p/2} \int\limits_{\alpha_{nk} R}^{\alpha_{nk}} dr~ r^{1-p/2}
= \frac{A_0^{p/2} [R^{2-p/2} - 1]}{p/2-2} ~ \alpha_{nk}^{2-p/2} ,
\end{equation*}
while the denominator can be simply bounded from below by a constant 
\begin{equation*}
\int\limits_{0}^{\alpha_{nk}} dr ~ r ~ |J_n(r)|^p \geq \int\limits_{0}^{1} dr ~ r ~ |J_n(r)|^p .
\end{equation*}
As a consequence, the ratio of $L_p$ norms in
Eq. (\ref{eq:focusingNorm}) goes to $0$ as $k$ increases.

For $1\leq p < 4$, one can write
\begin{equation*}
\frac{\| u_{nk} \|_{L_p(D(R))}^p}{\| u_{nk} \|_{L_p(D)}^p} = 1 - \frac{f_{p,n}(\alpha_{nk} R)}{f_{p,n}(\alpha_{nk})} ,
\end{equation*}
where
\begin{equation}
\label{eq:fp_def}
f_{p,n}(z) \equiv \int\limits_0^z dr ~ r ~ |J_n(r)|^p .
\end{equation}
As discussed in Remark \ref{theo:fp}, the function $f_{p,n}(z)$
behaves asymptotically as $z^{2-p/2}$ for large $z$, i.e., there
exists $0 < c_{p,n} < \infty$ such that for any $\ve > 0$, there
exists $z_0>0$ such that for any $z > z_0$ [cf. Eq. (\ref{eq:fp})]
\begin{equation*}
(c_{p,n} - \ve) z^{2-p/2} \leq f_{p,n}(z) \leq (c_{p,n} + \ve) z^{2-p/2} ,
\end{equation*}
from which one immediately deduces
\begin{equation*}
\frac{c_{p,n} - \ve}{c_{p,n} + \ve} R^{2-p/2} \leq \frac{f_{p,n}(\alpha_{nk} R)}{f_{p,n}(\alpha_{nk})} \leq \frac{c_{p,n} + \ve}{c_{p,n} - \ve} R^{2-p/2} .
\end{equation*}
As a consequence, for any $R < 1$, one can always choose $\ve$ such
that the right-hand side is strictly smaller than $1$ so that the
ratio of $L_p$ norms is then strictly positive.  Moreover, the
limiting value is $1 - R^{2-p/2}$ that completes the proof of
Eq. (\ref{eq:focusingNorm}) for $1\leq p < 4$. 
\end{proof}

\begin{remark}
\label{theo:fp}
For $1\leq p < 4$, the function $f_{p,n}(z)$ defined by
Eq. (\ref{eq:fp_def}) asymptotically behaves as $z^{2-p/2}$ for large
$z$, i.e., the limit
\begin{equation}
\label{eq:fp}
c_{p,n} = \lim\limits_{z\to\infty} \frac{f_{p,n}(z)}{z^{2-p/2}}
\end{equation}
exists, is finite and strictly positive: $0 < c_{p,n} < \infty$.
\end{remark}
Although this result is naturally expected from the asymptotic
behavior (\ref{eq:BesselapproxOne}) of Bessel functions, its rigorous
proof is beyond the scope of the paper.  An upper bound for the limit
(i.e., $c_{p,n} <\infty$) can be easily deduced from the inequality
(\ref{eq:Bessel_ineq}).  A lower strictly positive bound (i.e.,
$c_{p,n} > 0$) would require more careful estimations.  The most
difficult part consists in proving the existence of the limit, as the
numerical computation of the function $f_{p,n}(z)/z^{2-p/2}$ at large
$z$ shows its oscillatory behavior with a slowly decaying amplitude.
Since the statement of Theorem \ref{theo:focusingmodes} for $1 \leq p
< 4$ relies on this conjectural result, Eq. (\ref{eq:focusingNorm})
was also checked numerically and presented on
Fig. \ref{fig:ratio_focusing}.

\begin{figure}
\begin{center}
\includegraphics[width=60mm]{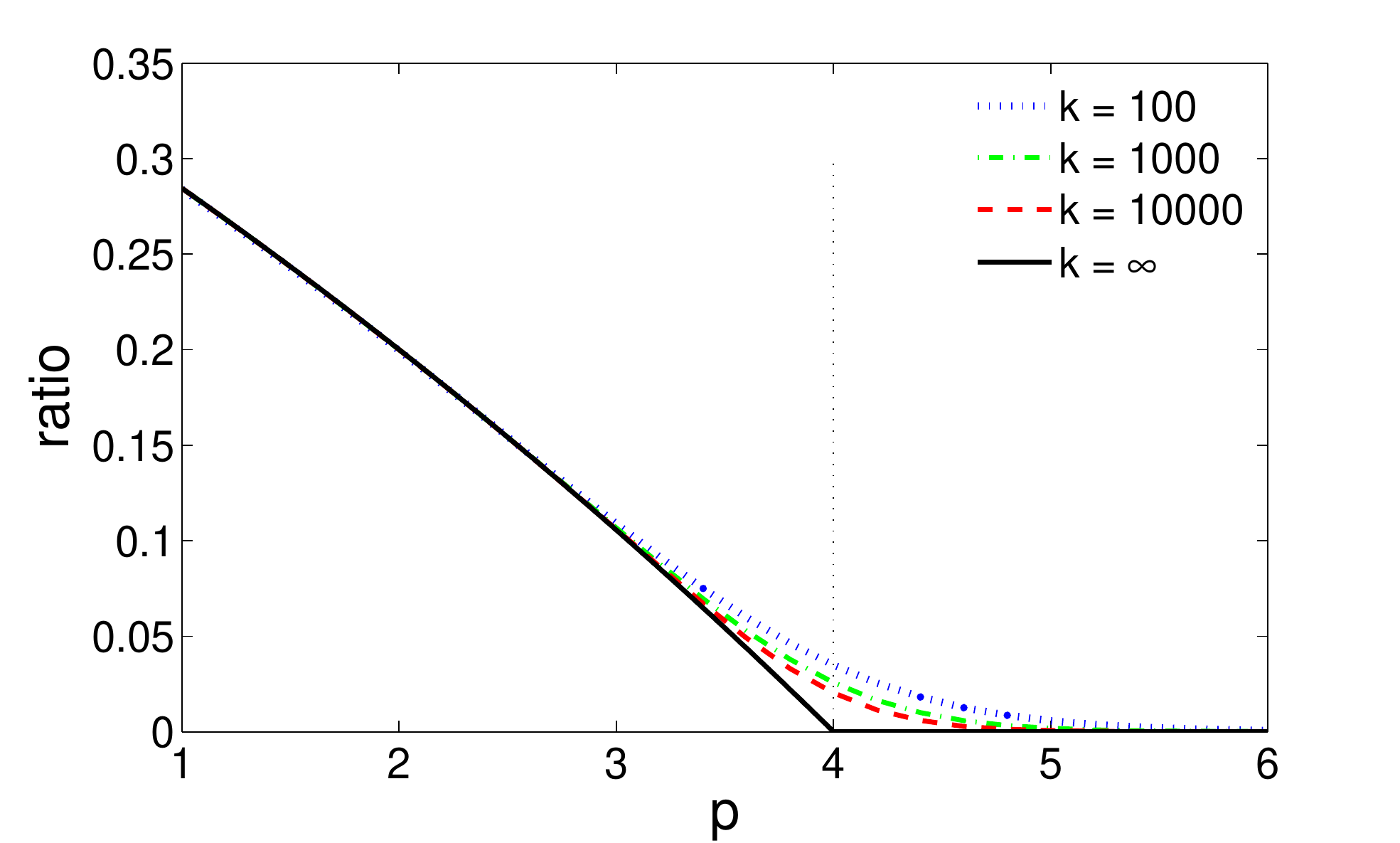}
\includegraphics[width=60mm]{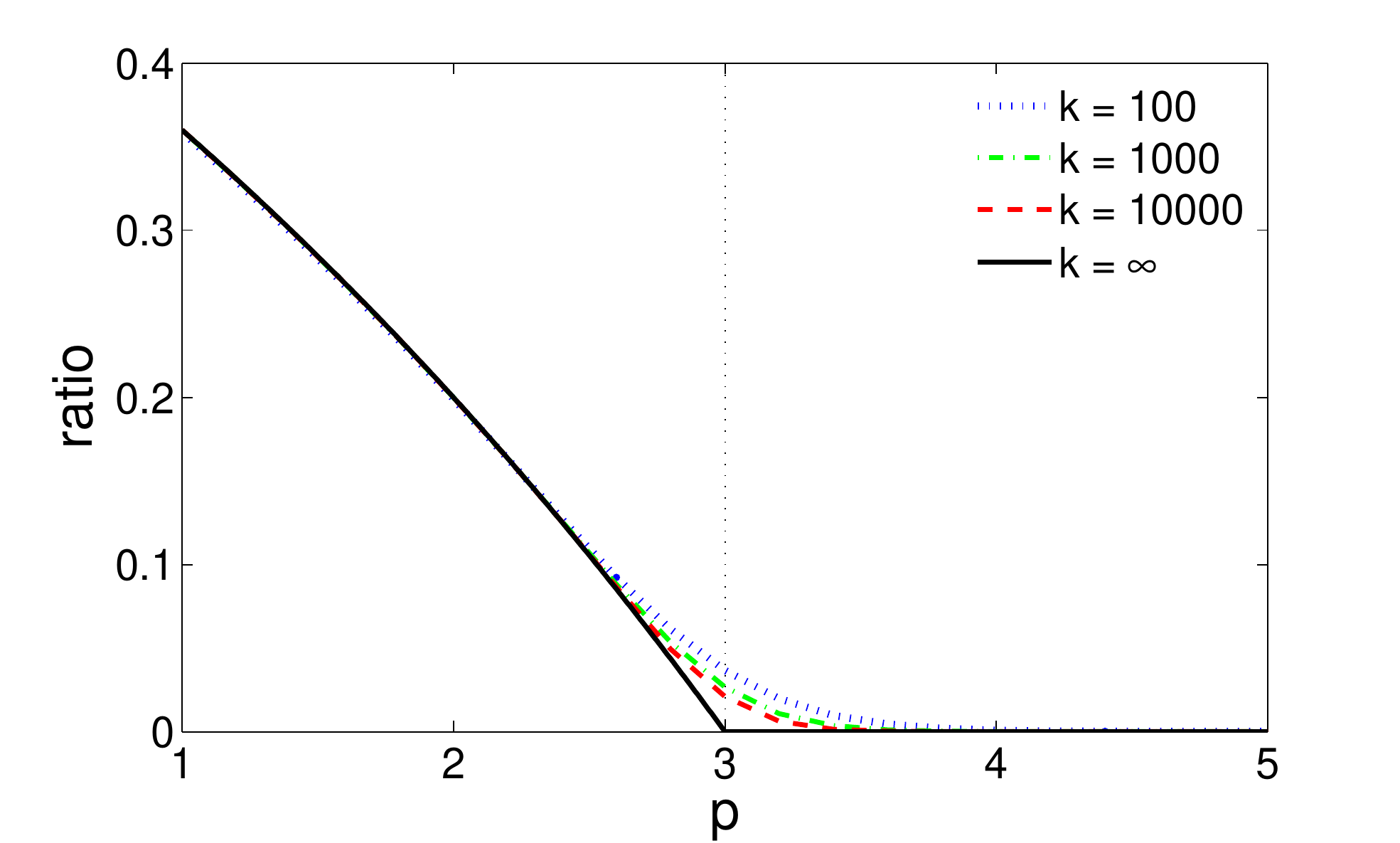}
\end{center}
\caption{
The ratio $\frac{\|u_{nk}\|_{L_p(D(R))}^p}{\|u_{nk}\|_{L_p(D)}^p}$ as a
function of $p$ for $n = 1$ and $R = 0.8$, in two dimensions (left)
and three dimensions (right).  Three color curves correspond to three
eigenfunctions with $k = 100$, $k = 1000$ and $k = 10000$, while the
solid black curve shows the theoretical limit as $k\to\infty$ given by
Eq. (\ref{eq:focusingNorm}) in 2D and Eq. (\ref{eq:focusingNorm_3d})
in 3D.  Note a slower convergence (deviations) for $p\sim 4$ in 2D and
$p\sim 3$ in 3D.} 
\label{fig:ratio_focusing}
\end{figure}

\section{Proofs for a ball}
\label{sec:Aball}

In this Appendix, we generalize the previous estimates to a ball $B =
\{ x\in\R^3~:~ |x| < R\}$ of radius $R > 0$.  We recall that the
Laplacian eigenfunctions in spherical coordinates are
\begin{equation}
\label{eq:u_ball}
u_{nkl}(r,\theta,\varphi) = j_n(\alpha_{nk} r/R) P_n^l(\cos\theta) e^{il\varphi},
\end{equation}
where $j_n(z)$ are the spherical Bessel functions of the first kind
(not to be confused with zeros $j_{n,k}$),
\begin{equation}
\label{eq:Bessels}
j_n(z) \equiv \sqrt{\frac{\pi}{2z}} J_{n+1/2}(z), 
\end{equation}
$P_n^l(z)$ are associated Legendre polynomials, and $\alpha_{nk}$ are
the positive zeros of $j_n(z)$ (Dirichlet), $j'_n(z)$ (Neumann) or
$j'_n(z) + h j_n(z)$ for some $h > 0$ (Robin), which are enumerated by
$k = 1,2,3,...$ for each $n = 0,1,2,...$.  We will derive the
estimates that do not depend on the angular coordinates $\theta$ and
$\varphi$, so that the last index $l$ will be omitted.

We start by recalling and extending several classical estimates.
\begin{lemma}
For any $\nu \in \R_+$ and any $x \in (0,1)$, the Kapteyn's inequality
holds \cite{Siegel53}
\begin{equation}
\label{eq:Kapteyn}
0 < J_{\nu}(\nu x) <  \frac{x^{\nu} \exp(\nu \sqrt{1-x^2})}{\left({1+ \sqrt{1-x^2}}\right)^{\nu}}.
\end{equation}
\end{lemma}

Now we can prove the following 
\begin{lemma}
For any $\nu > 1$ and $0<\epsilon<2/3$, one has
\begin{equation}
 0 < j_{\nu-\frac12}(x) < \sqrt{\frac{\pi}{2\nu}} ~ \exp\left(\frac{2}{3}- \frac{1}{3} \nu^{\epsilon}\right)
\qquad  \forall~ x \in (0, \nu - \nu^{\epsilon + 1/3}) .
\end{equation}
\label{lem:strongsphere}
\end{lemma}

\begin{proof}
Using the Kapteyn's inequality (\ref{eq:Kapteyn}) and taking $x = \nu
z$ with $z \in(0, 1 - \nu^{\epsilon-2/3})$, one has
\begin{equation*}
j_{\nu-1/2}(\nu z) = \sqrt{\frac{\pi}{2\nu}} \frac{J_{\nu}(\nu z)}{\sqrt{z}} < \sqrt{\frac{\pi}{2\nu}} f(z) ,
\hskip 10mm  \mathrm{with}  \hskip 5mm
 f(z) \equiv \frac{z^{\nu-1/2} e^{\nu\left({1-z^2}\right)^{1/2}}}{\left({1+ \sqrt{1-z^2}}\right)^{\nu}} .
\end{equation*}
Substituting $u=\sqrt{1-z^2} \in (0,1)$, one gets
\begin{equation*}
 f(z) = \left({\frac{\left({1-u^2}\right)^{1-\frac{1}{2\nu}} e^{2u}}{\left({1+u}\right)^2}}\right)^{\nu/2} = 
\frac{e^{\frac{u}{2}}}{\left({1+u}\right)^{\frac{1}{2}}} \left[{\left({\frac{1-u}{1+u}}\right)e^{2u}}\right]^{\frac{\nu}{2}-1/4} .
\end{equation*}
Using the inequality 
\begin{equation*}
 \frac{1-u}{1+u}~ e^{2u} < 1 - \frac{2}{3} u^3  \qquad \forall~ u \in (0,1) , 
\end{equation*}
one gets
\begin{equation*}
 f(z) < \frac{e^{\frac{1}{2}}}{\left({1+0}\right)^{\frac{1}{2}}} \left[{1 - \frac{2}{3} u^3}\right]^{\frac{\nu}{2}-1/4}
< e^{\frac{1}{2}}\left[{1 - \frac{2}{3} u^3}\right]^{\frac{\nu}{2}-1/4}.
\end{equation*}
Since $z < 1 - \nu^{\epsilon - 2/3}$, one has
\begin{equation*}
 u = \sqrt{1 - z^2} > \sqrt{1 - z} >  \nu ^{\frac{\epsilon}{2} - \frac{1}{3}} > \nu ^{\frac{\epsilon}{3} - \frac{1}{3}} ,
\end{equation*}
from which
\begin{equation*}
 f(z) < e^{\frac{1}{2}}\left[{1 - \frac{2}{3} \nu^{\epsilon-1}}\right]^{\frac{\nu}{2}-1/4} < 
e^{\frac{1}{2}}\left[{ \left({1 - \frac{2}{3} \nu^{\epsilon-1}}\right)^{\frac{3}{2}\nu^{1-\epsilon}}}
\right]^{\frac{\frac{\nu}{2}-1/4}{\frac{3}{2}\nu^{1-\epsilon}}} .
\end{equation*}
Since 
\begin{equation*}
  (1-x)^{\frac{1}{x}} < e^{-1}, \quad \forall  x \in (0,1), \quad \text{and} 
\quad 0 < \frac{2}{3} \nu^{\epsilon-1} < \frac{2}{3} <1,
\end{equation*}
one finally gets
\begin{equation*}
 f(z) < \exp\left(\frac{1}{2}-\frac{\frac{\nu}{2}-1/4}{\frac{3}{2}\nu^{1-\epsilon}}\right) < 
\exp\left(\frac{1}{2}+\frac{1}{6}\nu^{\epsilon-1} - \frac{1}{3} \nu^{\epsilon}\right) 
< \exp\left(\frac{2}{3}- \frac{1}{3} \nu^{\epsilon}\right) ,
\end{equation*}
that completes the proof.
\end{proof}

As a consequence, taking $\nu = n+1/2$ and $\epsilon = 1/3$, one has
\begin{lemma}
For $n = 1,2,...$ and any $z \in \bigl(0, n + 1/2 - (n + 1/2)^{2/3}\bigr)$,
\begin{equation}
\label{eq:jnz}
j_{n}(z) < \sqrt{\frac{\pi}{2n  + 1}} ~\exp\left(\frac{2}{3}- \frac{1}{3} \left({n  + \frac{1}{2}}\right)^{1/3}\right).
\end{equation}
\end{lemma}

The lemmas for Bessel functions and their zeros from Appendix
\ref{sec:Adisk} allow one to get similar estimates for spherical
Bessel functions $j_n(z)$, their positive zeros $\gamma_{n,k}$ and the
positive zeros $\gamma'_{n,k}$ of $j'_n(z)$.  They are summarized in
the following
\begin{lemma}
\label{lem:jn_3d}
For $n$ large enough, 
\begin{eqnarray}
\label{eq:jn_bound}
j_n(n+1/2) &>& \tilde{C}_1 (n+1/2)^{-5/6}    \qquad (\tilde{C}_1 = \sqrt{\pi/2}~ C_1),  \\
\label{eq:gamma_bound}
\gamma_{n,k} &<& (n+1/2) + \tilde{c}_k(n+1/2)^{1/3}  , \\
\label{eq:gamma_n}
\gamma'_{n,1} &>& n+1/2 + \tilde{C_2} (n+1/2)^{1/3}   \qquad (\tilde{C}_2 = 0.80) , \\
\label{eq:alpha_gamma}
\gamma'_{n,k} &<& \alpha_{nk} < \gamma_{n,k} .
\end{eqnarray}
\end{lemma}
\begin{proof}
From Lemma \ref{lem:root_estim}, we have
\begin{equation*}
j_{\nu-1/2}(\nu) = \sqrt{\frac{\pi}{2\nu}} J_\nu(\nu) > \sqrt{\frac{\pi}{2\nu}} C_1 \nu^{-1/3} = \tilde{C}_1 \nu^{-5/6} ,
\end{equation*}
from which (\ref{eq:jn_bound}) follows by taking $\nu = n+1/2$. \\
The zeros $\gamma_{n,k}$ of the spherical Bessel function $j_n(z)$ are
also the zeros of the Bessel function $J_{n+1/2}(z)$ so that
(\ref{eq:gamma_bound}) follows directly from Eq. (\ref{eq:j_nuk}) for
$\nu = n+1/2$ large enough.  \\
The inequality (\ref{eq:gamma_n}) follows from the asymptotic
expansion of $\gamma'_{n,1}$ for large $n$ \cite{Abramowitz} (p. 441)
\begin{equation*}
\begin{split}
\gamma'_{n,1} & = n+1/2 + 0.8086165 (n+1/2)^{1/3} - 0.236680 (n+1/2)^{-1/3} \\
& - 0.20736 (n+1/2)^{-1} + 0.0233 (n+1/2)^{-5/6} + ...  \qquad (n\gg 1). \\
\end{split}
\end{equation*}
Taking $\tilde{C_2} = 0.80$, one gets the inequality
(\ref{eq:gamma_n}). \\
Finally, the inequalities (\ref{eq:alpha_gamma}) follow from the
general minimax principle as for the disk.
\end{proof}

We also prove that the first maximum of the spherical Bessel function
at $\gamma'_{n,1}$ is the largest (although this is a classical fact,
we did not find an explicit reference).
\begin{lemma}
\label{lem:extrema_3d}
For an integer $n \ge 0$, one has
\begin{equation}
\max\limits_{x\in (0,\infty)} {j_n(x)} = j_n(\gamma'_{n,1}).
\end{equation}
\end{lemma}
\begin{proof}
The spherical Bessel function $j_n(x)$ satifies
\begin{equation*}
x^2 j''_n + 2x j'_n + [x^2- (n+1)n] j_n = 0.
\end{equation*}
Denoting $\kappa = n(n+1)$, one can rewrite this equation as
\begin{equation*}
j''_n(x) = -\frac{2x j'_n(x) + \left({x^2 - \kappa}\right) j_n(x)}{x^2},
\end{equation*}
from which
\begin{eqnarray*}
&& \frac{d}{dx} \left[{ \frac{x^2}{x^2-\kappa} \left({j'_n(x)}\right)^2 }\right] = 
\frac{d}{dx} \left[{ \left({j'_n(x)}\right)^2 +  \frac{\kappa}{x^2-\kappa}  \left({j'_n(x)}\right)^2 }\right] 
= 2j'_n(x) j''_n(x)  \\
&& + \kappa \left[{   \frac{2(x^2-\kappa)j'_n(x) j''_n(x) - 2x \left({j'_n(x)}\right)^2 }{\left({x^2 - \kappa}\right)^2} }\right]  
= \frac{2x^2}{x^2 - \kappa} j'_n(x) j''_n(x) - \frac{2x \kappa}{\left({x^2 - \kappa}\right)^2} \left({j'_n(x)}\right)^2  \\
&& = -\frac{2j'_n(x)}{x^2 - \kappa}  \left[{2x j'_n(x) + \left({x^2 - \kappa}\right) j_n(x)}\right] 
- \frac{2x \kappa}{\left({x^2 - \kappa}\right)^2} \left({j'_n(x)}\right)^2 = - 2j_n(x) j'_n(x) \\
&& - \frac{2x}{\left({x^2 - \kappa}\right)^2} \left({j'_n(x)}\right)^2 \left[{2(x^2-\kappa)+\kappa}\right]   
= - \frac{2x}{\left({x^2 - \kappa}\right)^2} \left({j'_n(x)}\right)^2 \left[{2x^2-\kappa}\right] - \frac{d}{dx} \left[j_n^2(x)\right] .
\end{eqnarray*}
Now, if we put
\begin{equation*}
\Lambda_n(x) = j_n^2(x) + \left[{ \frac{x^2}{x^2-\kappa} \left({j'_n(x)}\right)^2 }\right] ,
\end{equation*}
then
\begin{equation*}
\frac{d}{dx} \Lambda_n(x) = - \frac{2x}{\left({x^2 - \kappa}\right)^2} \left({j'_n(x)}\right)^2 \left[{2x^2-\kappa}\right] < 0 
\end{equation*}
for all $x > \sqrt{\frac{n(n+1)}{2}}$, i.e. $\Lambda_n(x)$
monotonously decreases.  Given that $\Lambda_n(\gamma'_{n,k}) =
j_n^2(\gamma_{n,k})$ and
\begin{equation*}
\sqrt{\frac{n(n+1)}{2}} < \gamma'_{n,1} < \gamma'_{n,2} < \dots  ,
\end{equation*}
we get the conclusion.
\end{proof}

Now, we can prove Theorem \ref{theo:ball}.
\begin{proof}
As earlier, the proof formalizes the idea that the eigenfunction
$u_{nk}$ is small in the large subdomain $B_{nk} = \{ x\in B ~:~ |x| <
R s_n/\alpha_{nk}\}$ (with $s_n = (n+1/2) - (n+1/2)^{2/3}$) and large
in the small subdomain $A_{nk} = \{ x\in B ~:~R (n+1/2)/\alpha_{nk} <
|x| < R \gamma'_{n,1}/\alpha_{nk}\}$.  Since $A_{nk}\subset B$, we
have for $1\leq p < \infty$
\begin{equation*}
\frac{\|u_{nk}\|_{L_p(B_{nk})}^p}{\|u_{nk}\|_{L_p(B)}^p} <
\frac{\|u_{nk}\|_{L_p(B_{nk})}^p}{\|u_{nk}\|_{L_p(A_{nk})}^p} = 
\frac{\int\limits_0^{Rs_n/\alpha_{nk}} dr ~r^2~ |j_n(r \alpha_{nk}/R)|^p}
{\int\limits_{R (n+1/2)/\alpha_{nk}}^{R \gamma'_{n,1}/\alpha_{nk}} dr~ r^2~ |j_n(r \alpha_{nk}/R)|^p} = 
\frac{\int\limits_0^{s_n} dz z^2~ |j_n(z)|^p}{\int\limits_{n+1/2}^{\gamma'_{n,1}} dz z^2~ |j_n(z)|^p} .
\end{equation*}
The numerator can be bounded by the inequality (\ref{eq:jnz}):
\begin{equation*}
\begin{split}
\int\limits_{0}^{s_n} dz z^2~ |j_n(z)|^p & < \left(\frac{\pi}{2n  + 1}\right)^{p/2} 
\exp\left(\frac{2p}{3}- \frac{p}{3} \left({n  + \frac{1}{2}}\right)^{1/3}\right) \frac{s_n^3}{3} \\
& <  \frac{(\pi/2)^{p/2}}{3} \exp\left(\frac{2p}{3}- \frac{p}{3} \left({n  + \frac{1}{2}}\right)^{1/3}\right) (n+1/2)^{3-p/2} \\
\end{split}
 (n=1,2,3,...).
\end{equation*}
In order to bound the denominator, we use the inequalities
(\ref{eq:jn_bound}, \ref{eq:gamma_n}) and the fact that $j_n(z)$
increases on the interval $[n+1/2,\gamma'_{n,1}]$ (up to the first
maximum at $\gamma'_{n,1}$):
\begin{equation*}
\begin{split}
& \int\limits_{n+1/2}^{\gamma'_{n,1}} dz z^2~ |j_n(z)|^p > [j_n(n+1/2)]^p \frac{(\gamma'_{n,1})^3 - (n+1/2)^3}{3}  \\
& >  [\tilde{C}_1 (n+1/2)^{-5/6}]^p ~ \frac{(n+1/2 + \tilde{C}_2(n+1/2)^{1/3})^3 - (n+1/2)^3}{3} 
 >  \tilde{C}_1^p \tilde{C}_2 (n+1/2)^{7/3 -5p/6}  \\
\end{split}
\end{equation*}
for $n$ large enough, from which
\begin{equation*}
\frac{\|u_{nk}\|_{L_p(B_{nk})}}{\|u_{nk}\|_{L_p(A_{nk})}} < \frac{\sqrt{\pi/2}}{\tilde{C}_1 (3\tilde{C}_2)^{1/p}} 
\exp\left(\frac{2}{3}- \frac{1}{3} \left({n  + \frac{1}{2}}\right)^{1/3}\right) (n+1/2)^{1/3 + 2/(3p)}     \qquad (n\gg 1)
\end{equation*}
that implies Eq. (\ref{eq:eigen_ball_L2}).  The case $p = \infty$
is treated similarly.
Finally, from Lemma \ref{lem:jn_3d}, we have for $n$ large enough
\begin{equation*}
1 > \frac{\mu_3(B_{nk})}{\mu_3(B)} = \left(\frac{s_n}{\alpha_{nk}}\right)^3 > \frac{s_n^3}{\gamma_{n,k}^3} > 
\frac{(n+1/2 - (n+1/2)^{2/3})^3}{(n+1/2 + \tilde{c}_k (n+1/2)^{1/3})^3}
\end{equation*}
so that the ratio of volumes tends to $1$ as $n$ goes to infinity.
\end{proof}

The proof of Theorem \ref{theo:focusingmodes_3d} for a ball is similar
to that of Theorem \ref{theo:focusingmodes}.
\begin{proof}
For $p = \infty$, the explicit representation (\ref{eq:u_ball}) of
eigenfunctions leads to
\begin{equation}
\label{eq:A_auxil1_3d}
 \frac{\|u_{nk}\|_{L_{\infty}(B(R))} }{ \|u_{nk}\|_{L_{\infty}(B)} }  
= \frac{\max\limits_{r\in[R,1]}{\left|j_n(\alpha_{nk} r)\right|}}{\max\limits_{r\in[0,1]}{\left|j_n(\alpha_{nk} r)\right|}}  
= \frac{\max\limits_{r\in[R,1]}{\left|j_n(\alpha_{nk} r)\right|}}{\left|j_n(\gamma'_{n,1})\right|} ,
\end{equation}
where we used the fact that the first maximum (at $\gamma'_{n,1}$) is
the largest (Lemma \ref{lem:extrema_3d}).  Since $\lim\limits_{k\to
\infty} \alpha_{nk} = \infty$, the spherical Bessel function
$j_{n}(\alpha_{nk} r)$ with $k\gg 1$ can be approximated in the
interval $[R,1]$ as \cite{Abramowitz}
\begin{equation}
\label{eq:BesselapproxOne_3d} 
j_n(\alpha_{nk} r) = \sqrt{\frac{\pi}{2\alpha_{nk} r}}~ J_{n+1/2}(\alpha_{nk} r) \approx 
\frac{1}{\alpha_{nk} r} \cos\left(\alpha_{nk} r - \frac{(n+1)\pi}{2}\right) .
\end{equation}
It means that there exists a positive integer $K_0$ and a constant
$A_0>0$ (e.g., $A_0 = 2$) such that
\begin{equation}
\label{eq:sBessel_ineq}
|j_n(\alpha_{nk} r)| < \frac{A_0}{\alpha_{nk} r} \leq \frac{A_0}{\alpha_{nk} R}, \quad \forall ~r\in[R,1], ~ k > K_0.
\end{equation}
Given that the denominator in Eq. (\ref{eq:A_auxil1_3d}) is fixed,
while the numerator decays as $\alpha_{nk}^{-1}$, one gets
Eq. (\ref{eq:focusingNorm_3d}).

For $p > 3$, the ratio of $L_p$ norms is 
\begin{equation}
\frac{\| u_{nk} \|_{L_p(D(R))}^p}{\| u_{nk} \|_{L_p(D)}^p} = 
\frac{\int\limits_{\alpha_{nk} R}^{\alpha_{nk}} dr ~ r^2 ~ |j_n(r)|^p}{\int\limits_{0}^{\alpha_{nk}} dr ~ r^2 ~ |j_n(r)|^p} .
\end{equation}
The inequality (\ref{eq:sBessel_ineq}) allows one to bound the
numerator as
\begin{equation*}
\int\limits_{\alpha_{nk} R}^{\alpha_{nk}} dr ~ r^2 ~ |j_n(r)|^p \leq A_0^{p} \int\limits_{\alpha_{nk} R}^{\alpha_{nk}} dz~ z^{2-p}
= \frac{A_0^p [R^{3-p} - 1]}{p-3} ~ \alpha_{nk}^{3-p} ,
\end{equation*}
while the denominator can be simply bounded from below by a constant 
\begin{equation*}
\int\limits_{0}^{\alpha_{nk}} dr ~ r^2 ~ |j_n(r)|^p \geq \int\limits_{0}^{1} dr ~ r^2 ~ |j_n(r)|^p .
\end{equation*}
As a consequence, the ratio of $L_p$ norms goes to $0$ as $k$ increases.

For $1\leq p < 3$, one can write
\begin{equation*}
\frac{\| u_{nk} \|_{L_p(B(R))}^p}{\| u_{nk} \|_{L_p(B)}^p} = 1 - \frac{\tilde{f}_{p,n}(\alpha_{nk} R)}{\tilde{f}_{p,n}(\alpha_{nk})} ,
\end{equation*}
where
\begin{equation}
\label{eq:fpt_def}
\tilde{f}_{p,n}(z) \equiv \int\limits_0^z dr ~ r^2 ~ |j_n(r)|^p .
\end{equation}
As discussed in Remark \ref{theo:fpt}, the function
$\tilde{f}_{p,n}(z)$ behaves asymptotically as $z^{3-p}$ for large
$z$, i.e., there exists $0 < \tilde{c}_{p,n} < \infty$ such that for
any $\ve > 0$, there exists $z_0>0$ such that for any $z > z_0$
[cf. Eq. (\ref{eq:fpt})]
\begin{equation*}
(\tilde{c}_{p,n} - \ve) z^{3-p} \leq \tilde{f}_{p,n}(z) \leq (\tilde{c}_{p,n} + \ve) z^{3-p} ,
\end{equation*}
from which one immediately deduces
\begin{equation*}
\frac{\tilde{c}_{p,n} - \ve}{\tilde{c}_{p,n} + \ve} R^{3-p} \leq \frac{\tilde{f}_{p,n}(\alpha_{nk} R)}{\tilde{f}_{p,n}(\alpha_{nk})} 
\leq \frac{\tilde{c}_{p,n} + \ve}{\tilde{c}_{p,n} - \ve} R^{3-p} .
\end{equation*}
As a consequence, for any $R < 1$, one can always choose $\ve$ such
that the right-hand side is strictly smaller than $1$ so that the
ratio of $L_p$ norms is then strictly positive.  Moreover, the
limiting value is $1 - R^{3-p}$ that completes the proof of
Eq. (\ref{eq:focusingNorm_3d}) for $1\leq p < 3$. 
\end{proof}

\begin{remark}
\label{theo:fpt}
The function $\tilde{f}_{p,n}(z)$ defined by Eq. (\ref{eq:fpt_def})
asymptotically behaves as $z^{3-p}$ for large $z$, i.e., the limit
\begin{equation}
\label{eq:fpt}
\tilde{c}_{p,n} = \lim\limits_{z\to\infty} \frac{\tilde{f}_{p,n}(z)}{z^{3-p}}
\end{equation}
exists, is finite and strictly positive: $0 < \tilde{c}_{p,n} < \infty$.
\end{remark}
As for Remark \ref{theo:fp}, a rigorous proof of this result is beyond
the scope of the paper.  The asymptotic behavior from
Eq. (\ref{eq:focusingNorm_3d}) for focusing modes is illustrated on
Fig. \ref{fig:ratio_focusing}.

\section{Analysis of Mathieu functions}
\label{sec:CalMathieu}

Many algorithms have been proposed for a numerical computation of
Mathieu functions \cite{Alhargen96,Shirts93,Vlasov92,Leeb79}.  The
main difficulty is the computation of Mathieu characteristic numbers
(MCNs).  Alhargan introduced a complete method for calculating the
MCNs for Mathieu functions of integer orders by using recurrence
relations for MCNs \cite{Alhargen96}.  His algorithm is a good
compromise between complexity, accuracy, speed and ease of use.
Nevertheless, for illustrative purposes of the paper, we used a
simpler approach by Zhang {\it et al.} \cite{Zhang96}.  In this
approach, the problem of calculating expansion coefficients of Mathieu
functions is reduced to an eigenproblem for sparse tridiagonal
matrices.  We have rebuilt the computation of Mathieu functions and
modified Mathieu functions and checked the accuracy of the numerical
algorithm by comparing their values to whose published in the
literature \cite{Chen94,Zhang96,Kirkpatrick60}.  We also checked that
the truncation of the underlying tridiagonal matrices to the size
$K_{max} = 200$ was enough for getting very accurate results, at least
for the examples presented in the paper.

\section{Asymptotic behavior of Mathieu functions for large $q$}
\label{sec:Aelliptical}

The large $q$ asymptotic expansions of $\ce_n(z,q)$ and
$\se_{n+1}(z,q)$ for $z \in [0,\frac{\pi}{2})$ and $n = 0,1,2,...$ are
\cite{Mclachlan47,Frenkel2001}
\begin{eqnarray}
\label{eq:asympExpanCosinMathieu}
\ce_n(z,q) &=& C_n(q) \left({   e^{2\sqrt{q}\sin z} h^+_n(z) \sum\limits_{k=0}^{\infty} \frac{f^{+}_k(z)}{q^{k/2}} + 
e^{-2\sqrt{q}\sin z} h^-_n(z) \sum\limits_{k=0}^{\infty}  \frac{f^{-}_k(z)}{q^{k/2}}  }\right), \\
\label{eq:asympExpanSinMathieu}
\se_{n+1}(z,q) &=& S_{n+1}(q) \left({   e^{2\sqrt{q}\sin z} h^+_n(z) \sum\limits_{k=0}^{\infty} \frac{f^{+}_k(z)}{q^{k/2}} 
- e^{-2\sqrt{q}\sin z} h^-_n(z) \sum\limits_{k=0}^{\infty}  \frac{f^{-}_k(z)}{q^{k/2}}   }\right) ,
\end{eqnarray}
where 
\begin{eqnarray}
\label{eq:A_h1n}
h^+_n(z) &=& 2^{n+\frac{1}{2}} \frac{\left[{\cos \left({\frac{1}{2}z + \frac{\pi}{4}}\right)}\right]^{2n+1}}{\left({\cos z}\right)^{n+1}}
= \sqrt{ \frac{(1 - \sin z)^n}{(1 + \sin z)^{n+1}} } ,  \\
\label{eq:A_h2n}
h^-_n(z) &=& 2^{n+\frac{1}{2}} \frac{\left[{\sin \left({\frac{1}{2}z + \frac{\pi}{4}}\right)}\right]^{2n+1}}{\left({\cos z}\right)^{n+1}}
= \sqrt{ \frac{(1 + \sin z)^n}{(1 - \sin z)^{n+1}} } , 
\end{eqnarray}
and the coefficients $C_n(q)$ and $S_{n+1}(q)$ are given explicitly in
\cite{Mclachlan47}.
The coefficients $f^\pm_k(z)$ can be computed through the recursive
formulas given in \cite{Frenkel2001}, e.g.
\begin{equation*}
f^{\pm}_0(z) =1,  \qquad   f^{\pm}_1(z) = \frac{2n+1 \mp \left({n^2+n+1}\right)\sin{z}}{8\cos(z)^2}.
\end{equation*}
When $q$ is large enough, one can truncate the asymptotic expansions
(\ref{eq:asympExpanCosinMathieu}, \ref{eq:asympExpanSinMathieu}) by
keeping only two terms ($k = 0,1$) and get accurate approximations for
$\ce_n$ and $\se_{n+1}$, as illustrated on
Fig. \ref{fig:check_ce_m_numerical_vs_Goldstein_expansion}.

\begin{figure}
\begin{center}
\includegraphics[width=60mm]{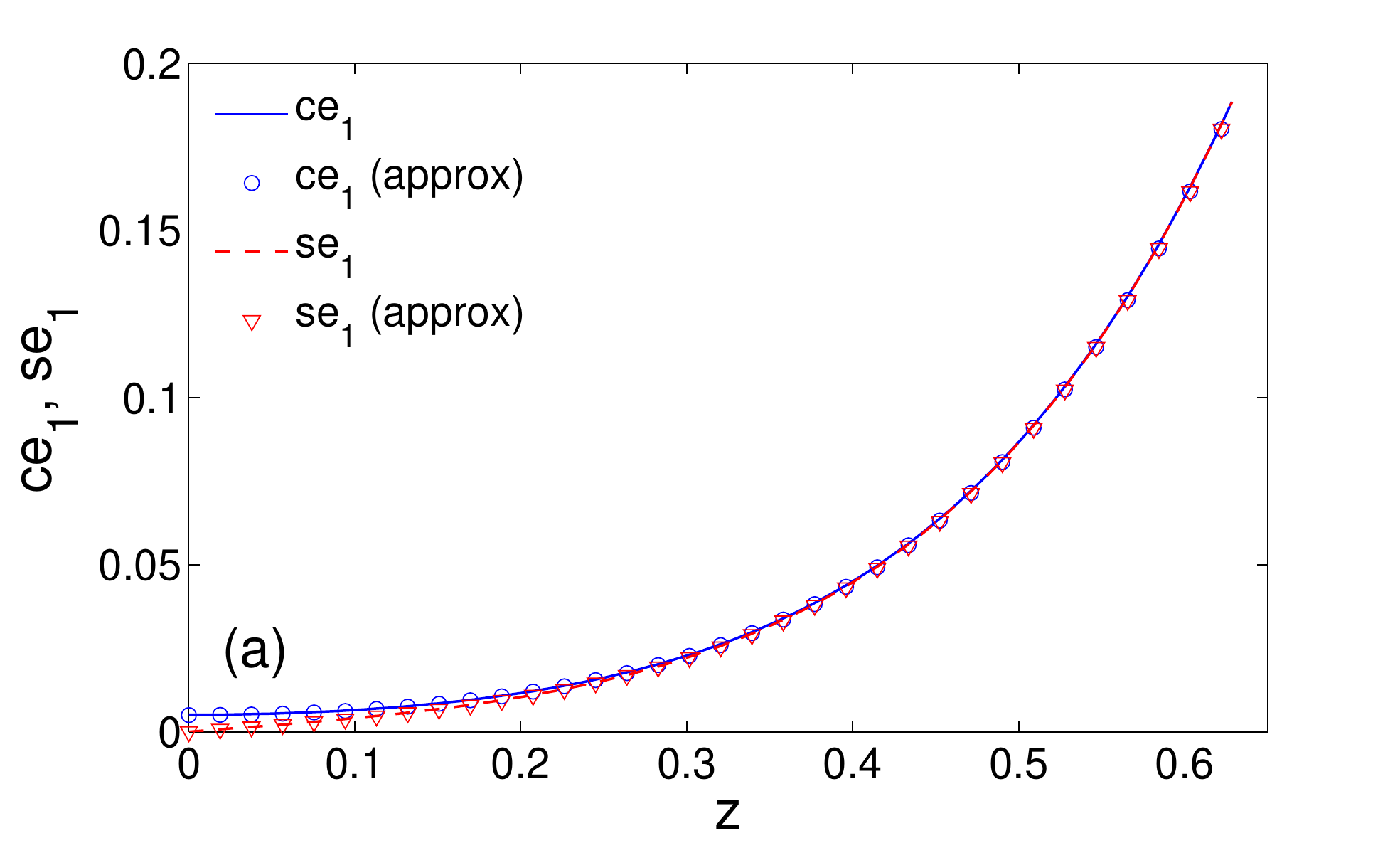}
\end{center}
\caption{
The functions $\ce_1(z,q)$ and $\se_1(z,q)$ (solid and
dashed lines), computed by our algorithm with $K_{max} = 200$, and
their approximations (circles and triangles) by the asymptotic
expansions (\ref{eq:asympExpanCosinMathieu},
\ref{eq:asympExpanSinMathieu}) truncated to two terms ($k = 0,1$),
with $q = 20$. }
\label{fig:check_ce_m_numerical_vs_Goldstein_expansion}
\label{fig:check_L_1n_gamma_q}
\end{figure}

It is convenient to define the functions
\begin{equation*}
G_n^\pm(z,q) = h^+_n(z) \pm e^{-4\sqrt{q} \sin(z)} h^-_n(z) +  h^+_n(z) \sum\limits_{k=1}^{\infty} 
 \frac{f^{+}_k(z)}{q^{k/2}} \pm  e^{-4\sqrt{q} \sin(z)} h^-_n(z) \sum\limits_{k=1}^{\infty}  \frac{f^{-}_k(z)}{q^{k/2}} , 
\end{equation*}
in order to write
\begin{equation*}
\ce_n(z,q) = C_n(q) e^{2\sqrt{q} \sin(z)} G_n^+(z,q) ,  \qquad  \se_{n+1}(z,q) = S_{n+1}(q) e^{2\sqrt{q} \sin(z)} G_n^-(z,q) .
\end{equation*}
In what follows, we will estimate the functions $G_n^\pm(z,q)$ by
their leading terms, given that the remaining part is getting small
for large $q$.

\begin{lemma}
\label{lem:bound}
For $\gamma\in \left(0, \frac{\pi}{2}\right)$, there exists $N_\gamma
> 0$ such that for $q > N_\gamma$ and $z\in (0,\gamma)$
\begin{equation}
\left|\sum\limits_{k=1}^{\infty} \frac{f^{\pm}_k(z)}{q^{k/2}} \right|  < \frac12 .
\end{equation}
\end{lemma}

Now, we establish the upper and lower bounds for the functions
$G_n^\pm$.

\begin{lemma}
\label{lem:localizedfunctionGm}
Let $\alpha \in \left(0,\frac{\pi}{2}\right)$, $\gamma \in
\left(\alpha,\frac{\pi}{2}\right)$.  Then, there exists
$N_\gamma >0$ such that for any $\beta \in (\alpha,\gamma)$ and
$q>N_\gamma$:
\begin{eqnarray}
\left|G_n^\pm(z_1,q)\right| &<& \frac{3}{2} \left(1 + h^-_n(\alpha) e^{-4\sqrt{q} \sin(z_1)}\right) \qquad \forall ~z_1 \in (0,\alpha),  \\
\left|G_n^\pm(z_2,q)\right| &>& \frac{1}{2} h^+_n(\gamma)     \qquad   \forall ~z_2 \in (\beta,\gamma) . 
\end{eqnarray}
\end{lemma}
\begin{proof}
From Lemma \ref{lem:bound}, there exists $N_\gamma > 0$ such that for
$q> N_\gamma$ and $z_1 \in (0,\alpha)$, one has
\begin{equation*}
\left|G_n^\pm(z_1,q)\right|  < \frac{3}{2} \left({h^+_n(z_1) + h^-_n(z_1) e^{-4\sqrt{q} \sin(z_1)}}\right)
< \frac{3}{2} \left(1 + h^-_n(\alpha) e^{-4\sqrt{q} \sin(z_1)}\right) . 
\end{equation*}

For $q> N_\gamma$ and $z_2 \in (\beta,\gamma)$, one has
\begin{equation*}
\left|G_n^+(z_2,q)\right| 
 > \frac{1}{2} \left({h^+_n(z_2) + h^-_n(z_2) e^{-4\sqrt{q} \sin(z_2)}}\right)
> \frac{1}{2} h^+_n(\gamma) > 0
\end{equation*}
and
\begin{eqnarray*}
\left|G_n^-(z_2,q) - \left(h^+_n(z_2) - h^-_n(z_2) e^{-4\sqrt{q} \sin(z_2)}\right)\right| 
< \frac{1}{2} \left(h^+_n(z_2) + h^-_n(z_2) e^{-4\sqrt{q} \sin(z_2)}\right) .
\end{eqnarray*}
The last inequality implies
\begin{equation*}
\left|G_n^-(z_2,q)\right| > \min \left\{ \left(h^+_n(z_2) - h^-_n(z_2) e^{-4\sqrt{q} \sin(z_2)}\right),
\frac12 \left(h^+_n(z_2) - 3 h^-_n(z_2) e^{-4\sqrt{q} \sin(z_2)}\right) \right\}  .
\end{equation*}
Since $h^-_n(z_2) > 0$ and $h^+_n(z_2)$ is a decreasing function, one
gets
\begin{equation*}
\left|G_n^-(z_2,q)\right| > \frac12 h^+_n(\gamma) ,
\end{equation*}
that completes the proof.
\end{proof}

Now we can prove Theorem \ref{theo:localizationEllipticalAnnuliMathieuEigenf}.
\begin{proof}
We first consider the case $i=1$.  Using the symmetric properties of
Mathieu functions \cite{Zhang96}, one has
\begin{equation*}
\frac{\left\|u_{nk1}\right\|^p_{L_p(\Omega \setminus \Omega_{\alpha})}}{\left\|{u_{nk1}}\right\|^p_{L_p(\Omega_\alpha)}} = 
\frac{\int\limits_0^{\alpha} |\ce_n(z_1,q_{nk1})|^p dz_1}{\int\limits_{\alpha}^{\pi/2} |\ce_n(z_2,q_{nk1})|^p dz_2}.
\end{equation*}
Choosing $\beta = \frac{\pi}{4}+\frac{\alpha}{2}$ and $\gamma=
\frac{3\pi}{8}+\frac{\alpha}{4}$, one gets
\begin{equation*}
\frac{\int\limits_0^{\alpha} |\ce_n(z_1,q_{nk1})|^p dz_1}{\int\limits_{\alpha}^{\pi/2} |\ce_n(z_2,q_{nk1})|^p dz_2} < 
\frac{\int\limits_0^{\alpha} |\ce_n(z_1,q_{nk1})|^p dz_1}{\int\limits_{\beta}^{\gamma} |\ce_n(z_2,q_{nk1})|^p dz_2} .
\end{equation*}
From Lemma \ref{lem:localizedfunctionGm}, there exists $N_{\gamma} >0$
such that for $q > N_\gamma$,
\begin{eqnarray*}
\int\limits_0^{\alpha} && |\ce_n(z_1,q)|^p dz_1 = (C_n(q))^p \int\limits_0^{\alpha} e^{2p\sqrt{q}\sin z_1} |G_n^+(z_1,q)|^p dz_1 \\
&<& (C_n(q))^p \left(\frac32\right)^p 
\int\limits_0^\alpha \left(\sum\limits_{k=0}^p \binom{p}{k} [e^{2\sqrt{q}\sin z_1}]^{p-2k} (h^-_n(\alpha))^k \right) dz_1  \\
&\leq& \alpha (C_n(q))^p \left(\frac32\right)^p \left(\sum\limits_{k=0}^{[p/2]} \binom{p}{k} [e^{2\sqrt{q}\sin \alpha}]^{p-2k} (h^-_n(\alpha))^k
+ \sum\limits_{k=[p/2]+1}^{p} \binom{p}{k} (h^-_n(\alpha))^k \right) ,  \\
\end{eqnarray*}
where the terms $e^{m\sqrt{q}\sin z_1}$ were bounded by
$e^{m\sqrt{q}\sin \alpha}$ for $m > 0$, and by $1$ for $m \leq 0$
(here $[x]$ denotes the integer part of $x$).  In addition,
\begin{eqnarray*}
\int\limits_{\beta}^{\gamma} |\ce_n(z_2,q)|^p dz_2 &>& (C_n(q))^p\left(\frac12\right)^p (h^+_n(\gamma))^p
\int\limits_{\beta}^{\gamma} e^{2p\sqrt{q}\sin z_2} dz_2 \\
&>& (C_n(q))^p\left(\frac12\right)^p (h^+_n(\gamma))^p  (\gamma-\beta) e^{2p\sqrt{q}\sin \beta} , 
\end{eqnarray*}
from which
\begin{eqnarray*}
&& \frac{\left\|u_{nk1}\right\|^p_{L_p(\Omega \setminus \Omega_{\alpha})}}{\left\|{u_{nk1}}\right\|^p_{L_p(\Omega_\alpha)}} <
\frac{3^p \alpha}{(\gamma-\beta) (h^+_n(\gamma))^p} e^{-2p\sqrt{q_{nk1}}(\sin\beta - \sin\alpha)} \Biggl(1 + \\
&& \left[\sum\limits_{k=1}^{[p/2]} \binom{p}{k} (e^{2\sqrt{q}\sin \alpha})^{-2k} (h^-_n(\alpha))^k + e^{-2p\sqrt{q}\sin\alpha}
\sum\limits_{k=[p/2]+1}^{p} \binom{p}{k} (h^-_n(\alpha))^k \right]\Biggr) .
\end{eqnarray*}
Taking $q$ large enough, one can make the terms in large brackets
smaller than any prescribed threshold $\epsilon$.  For $\epsilon = 1$,
one can simplify the estimate as
\begin{equation*}
\frac{\left\|u_{nk1}\right\|^p_{L_p(\Omega \setminus \Omega_{\alpha})}}{\left\|{u_{nk1}}\right\|^p_{L_p(\Omega_\alpha)}} < 
2 \frac{3^p \alpha}{(\gamma-\beta) (h^+_n(\gamma))^p} \exp\biggl[-2p\sqrt{q_{nk1}}(\sin\beta - \sin\alpha)\biggr] .
\end{equation*}
Substituting $\beta = \pi/4 + \alpha/2$ and $\gamma = 3\pi/8 +
\alpha/4$, one gets Eq. (\ref{eq:bound_elliptical}) after
trigonometric simplifications.

For $i=2$, one can use similar estimates for $\se_{n+1}$.
\end{proof}

\section{No localization in rectangle-like domains}
\label{sec:Arectangle}

Theorem \ref{theo:rectangle} relies on the following simple estimate.
\begin{lemma}
For $0 \leq a < b$ and any positive integer $m$, one has
\begin{equation}
\label{eq:sin2}
\begin{split}
\int\limits_{a}^{b} |\sin(m x)| dx \geq \int\limits_{a}^{b} \sin^2(m x) dx \geq \epsilon(a,b) & > 0 , \\ 
\int\limits_{a}^{b} |\cos(m x)| dx \geq \int\limits_{a}^{b} \cos^2(m x) dx \geq \epsilon(a,b) & > 0 , \\
\end{split}
\end{equation}
where
\begin{equation}
\epsilon(a,b) = \min\left\{\frac{b-a}{4},  \frac{b-a}{2} - \frac{1}{2} \left|\frac{\sin(n(b-a))}{n}\right| : n=1,2,\dots, 
\left[{\frac{2}{b-a}}\right]\right\} > 0,
\end{equation}
\end{lemma}
It is important to stress that the lower bound $\epsilon(a,b)$ does
not depend on $m$.  The proof of this lemma is elementary.

The proof of Theorem \ref{theo:rectangle} is a simple consequence.
\begin{proof}
The condition (\ref{eq:rect_cond}) ensures that all the eigenvalues
are simple so that each eigenfunction is 
\begin{equation*}
u_{n_1,...,n_d}(x_1,...,x_d) =
\begin{cases}
\sin(\pi n_1 x_1/\ell_1) ... \sin(\pi n_d x_d/\ell_d) \qquad \rm{(Dirichlet)}, \cr
\cos(\pi n_1 x_1/\ell_1) ... \cos(\pi n_d x_d/\ell_d) \qquad \rm{(Neumann)} . \end{cases}
\end{equation*}
For any open subset $V$, there exists a ball included in $V$ and thus
there exists a rectangle-like domain $\Omega_V = [a_1,b_1]\times
... \times [a_d,b_d]\subset V$, with $0 \leq a_i < b_i \leq \ell_i$
for all $i=1,...,d$.  The $L_1$-norm of $u$ in $V$ can be estimated as
\begin{equation*}
\begin{split}
& \|u_{n_1,...,n_d}\|_{L_1(V)} \geq \|u_{n_1,...,n_d}\|_{L_1(\Omega_V)} = 
\prod\limits_{i=1}^d \int\limits_{a_i}^{b_i} dx_i \begin{cases} |\sin(\pi n_i x_i/\ell_i)|  \cr |\cos(\pi n_i x_i/\ell_i)| \end{cases} \\
& = \frac{\ell_1 ... \ell_d}{\pi^d} \prod\limits_{i=1}^d \int\limits_{\pi a_i/\ell_i}^{\pi b_i/\ell_i} dx_i 
\begin{cases} |\sin(n_i x_i)| \cr  |\cos(n_i x_i)| \end{cases}
\geq \frac{\ell_1 ... \ell_d}{\pi^d} \prod\limits_{i=1}^d \epsilon(\pi a_i/\ell_i, \pi b_i/\ell_i)  ,\\
\end{split}
\end{equation*}
where the last inequality results from (\ref{eq:sin2}).  To complete
the proof, one uses the Jensen's inequality for $L_p$-norms and
$\mu_d(V) \geq \mu_d(\Omega_V) = (b_1-a_1)...(b_d-a_d)$
\begin{equation*}
\begin{split}
\frac{\|u_{n_1,...,n_d}\|_{L_p(V)}}{\|u_{n_1,...,n_d}\|_{L_p(\Omega)}} & >
\frac{\|u_{n_1,...,n_d}\|_{L_1(V)} (\mu_d(V))^{\frac{1}{p}-1}}{\|u_{n_1,...,n_d}\|_{L_\infty(\Omega)} (\mu_d(\Omega))^{\frac{1}{p}}} 
 \geq \frac{1}{\pi^d} \prod\limits_{i=1}^d \left(\frac{b_i - a_i}{\ell_i}\right)^{\frac{1}{p}-1} 
\epsilon(\pi \frac{a_i}{\ell_i}, \pi \frac{b_i}{\ell_i}) > 0 . \\
\end{split}
\end{equation*}
Since the right-hand side is strictly positive and independent of
$n_1$, ... , $n_d$, the infimum of the left-hand side over all
eigenfunctions is strictly positive.
\end{proof}

\end{document}